\documentclass[11pt,letterpaper]{article}
\usepackage[margin=1in]{geometry}
\usepackage[T1]{fontenc}
\usepackage{lmodern}
\usepackage{amsmath,amssymb,amsthm,mathtools}
\usepackage{microtype,enumitem}
\usepackage[colorlinks=true,linkcolor=blue,citecolor=blue,urlcolor=blue]{hyperref}
\hypersetup{pdftitle={Hardness Amplification for Approximate Nash Equilibrium with Constantly Many Players},pdfauthor={}}
\newtheorem{theorem}{Theorem}[section]
\newtheorem{lemma}[theorem]{Lemma}
\newtheorem{proposition}[theorem]{Proposition}
\newtheorem{corollary}[theorem]{Corollary}
\newcommand{\eps}{\varepsilon}
\newcommand{\E}{\mathbb E}
\newcommand{\Prb}{\mathbb P}
\newcommand{\ind}{\mathbf 1}

\newcommand{\Filter}{\mathsf{Filter}}
\DeclareMathOperator*{\argmax}{arg\,max}
\numberwithin{equation}{section}
\setlist{topsep=4pt,itemsep=2pt,parsep=0pt}
\title{Hardness Amplification for Approximate Nash Equilibrium\\with Constantly Many Players}

\usepackage{authblk}
\author[1,2]{Xiaotie Deng}
\author[1]{Ningyuan Li}
\author[3]{Zihe Wang}
\author[4]{Jie Zhang}
\affil[1]{Center on Frontiers of Computing Studies, School of Computer Science, Peking University\\
\texttt{liningyuan@pku.edu.cn}}
\affil[2]{Department of Computer Science, City University of Hong Kong\\
\texttt{csdeng@cityu.edu.hk}}
\affil[3]{Gaoling School of Artificial Intelligence, Renmin University of China\\
\texttt{wang.zihe@ruc.edu.cn}}
\affil[4]{Department of Computer Science, University of Bath\\
\texttt{jz2558@bath.ac.uk}}
\date{}
\begin{document}
\begin{titlepage}
\maketitle
\thispagestyle{empty}
\begin{abstract}
\normalsize
For every fixed integer $k\ge2$ and fixed parameters $\delta,\eps\in(0,1)$ and $\eta\in[0,1)$, we give a deterministic polynomial-time reduction from finding a $\delta$-approximate Nash equilibrium ($\delta$-ANE) in a $k$-player game $G$ with payoffs in $[0,1]$ to finding an $(\eps,\eta)$-weak approximate Nash equilibrium (weak ANE) in a game $H$ with $C$ players, where $C$ depends only on $k,\delta,\eps,\eta$, independently of the action-set sizes. The constructed game is symmetric with binary payoffs. 
The games and strategies are represented by complete payoff tables and probabilities, with rational entries encoded in binary. 

Consequently, the existence of deterministic polynomial-time weak-ANE algorithms at one fixed pair of parameters $\eps\in(0,1),\eta\in[0,1)$ for every fixed number of players would give deterministic polynomial-time $\delta$-ANE algorithms at every fixed approximation parameter $\delta\in(0,1)$ and number of players.

Under deterministic PPAD-ETH, combining our reduction with Rubinstein's hardness result (FOCS~2016) for bimatrix games at some fixed approximation parameter $\gamma_0\in(0,1)$ rules out a deterministic algorithm for $(\eps,\eta)$-weak ANE with running time $N^{O((\log N)^{1-\xi})}$ for any fixed $\xi\in(0,1)$. This holds for every fixed $\eps\in(0,1)$ and $\eta\in[0,1)$ at a suitable constant number $C$ of players, where $N$ is the binary encoding length of the constructed game's complete payoff table.

\end{abstract}
\end{titlepage}
\setcounter{page}{1}

\section{Introduction}\label{sec:introduction}

We consider finite normal-form games with payoffs in $[0,1]$, represented by complete payoff tables. An $\eps$-approximate Nash equilibrium ($\eps$-ANE) is a mixed strategy profile at which no player can improve its expected payoff by more than $\eps$ through a unilateral deviation. Every mixed strategy profile is a $1$-ANE. Unless stated otherwise, our computational statements concern deterministic algorithms.

Finding an ANE with an inverse-polynomial approximation parameter is known to be PPAD-complete~\cite{DBLP:journals/siamcomp/DaskalakisGP09,DBLP:journals/jacm/ChenDT09}.
Under deterministic PPAD-ETH, Rubinstein~\cite{DBLP:conf/focs/Rubinstein16} shows that finding a $\gamma_0$-ANE in bimatrix games requires quasipolynomial time for some fixed $\gamma_0\in(0,1)$. This rules out a polynomial-time approximation scheme even in two-player games, but leaves open the possibility of polynomial-time algorithms at larger approximation parameters.

On the algorithmic side, constant approximation guarantees are available for games with few players.
A line of work on polynomial-time approximation algorithms for bimatrix games has led to progressively better guarantees. Li and Li~\cite{LL26} recently report a guarantee of approximately $0.30954$. 
For games with a fixed number $k\ge2$ of players, polynomial-time algorithms compute a $(1-\frac1k)$-ANE~\cite{DBLP:journals/corr/abs-0804-4524}. Extensions based on better bimatrix algorithms improve this guarantee, but the resulting bounds still approach one as $k$ grows~\cite{DBLP:conf/sagt/HemonRS08}. Deligkas, Fearnley, Savani, and Spirakis~\cite{DBLP:journals/algorithmica/DeligkasFSS17} noted that no fixed approximation guarantee below one was known to be achievable in polynomial time for every fixed number of players.
This raises the following question:
\begin{quote}
\emph{Does there exist a constant $\eps\in(0,1)$ such that, for every fixed $k\ge2$, an $\eps$-ANE of any $k$-player game can be computed in polynomial time?}
\footnote{Throughout the paper, the running time is measured in the binary encoding length of the complete payoff table, and the polynomial degree may depend on the fixed player number $k$. A stronger requirement is a degree independent of $k$, and our negative answer also rules out that possibility.}
\end{quote}

In this paper, we give a conditional negative answer to this question. 
We amplify the hardness of approximate Nash equilibrium from any fixed approximation parameter $\delta\in(0,1)$ to every fixed approximation parameter $\eps\in(0,1)$.
For every fixed $k\ge2$ and $\delta,\eps\in(0,1)$, we give an unconditional polynomial-time reduction from $\delta$-ANE in $k$-player games to $\eps$-ANE in games with $C$ players, where $C$ depends only on $k,\delta,\eps$. Moreover, the constructed games are symmetric with binary payoffs. Under PPAD-ETH, combining this reduction with Rubinstein's hardness result~\cite{DBLP:conf/focs/Rubinstein16} yields
quasipolynomial hardness for every fixed $\eps\in(0,1)$. 

Our reduction further shows that the conditional hardness  persists for weak approximate Nash equilibria, where at least a positive fraction of the players must satisfy the approximation requirement. Following Babichenko, Papadimitriou, and Rubinstein~\cite[Definition~4]{DBLP:conf/innovations/BabichenkoPR16}, a mixed strategy profile in a $p$-player game is an $(\eps,\eta)$-\emph{weak approximate Nash equilibrium}, or $(\eps,\eta)$-weak ANE, if at most $\lfloor\eta p\rfloor$ players can improve their expected payoff by more than $\eps$. The case $\eta=0$ gives the usual $\eps$-ANE. Our reduction applies to this relaxation for every fixed $\eps\in(0,1)$ and $\eta\in[0,1)$, with a suitable constant number of players.

\subsection{Main result and consequences}\label{subsec:results}

The reduction constructs a game $H$ whose weak approximate equilibria allow recovery of an approximate equilibrium of the input game $G$. Write $M$ for the encoding length of $G$ and $N$ for that of $H$. Payoffs and probabilities are rational numbers encoded in binary.

\begin{theorem}[Construction and recovery]\label{thm:amplification}
Fix an integer $k\ge2$ and parameters $\delta,\eps\in(0,1)$ and $\eta\in[0,1)$.
There exist an integer $C=C(k,\delta,\eps,\eta)\ge2$ and deterministic algorithms with the following  guarantees.
\begin{enumerate}[label=\textnormal{(\roman*)},leftmargin=*]
\item Given a $k$-player game $G$ with payoffs in $[0,1]$, the construction outputs the complete payoff table of a symmetric $C$-player binary-payoff game $H$ in time polynomial in the encoding length $M$ of $G$.
\item From any $(\eps,\eta)$-weak ANE $\sigma$ of $H$, the recovery calculation computes a $\delta$-ANE of $G$ in time polynomial in $M+L_\sigma$, where $L_\sigma$ is the actual encoding length of the supplied profile.
\end{enumerate}
The polynomial degrees depend only on $k,\delta,\eps,\eta$.
\end{theorem}

The following corollary makes the algorithmic equivalence precise. 

\begin{corollary}[Equivalence of polynomial-time tractability]\label{cor:polytime}
The following statements are equivalent.
\begin{enumerate}[label=\textnormal{(\roman*)},leftmargin=*]
\item There exist fixed $\eps\in(0,1)$ and $\eta\in[0,1)$ such that, for every fixed $C\ge2$, an $(\eps,\eta)$-weak ANE of a symmetric $C$-player binary-payoff game can be found deterministically in polynomial time.
\item For every fixed $\delta\in(0,1)$ and every fixed $k\ge2$, a $\delta$-ANE of an arbitrary $k$-player game with payoffs in $[0,1]$ can be found deterministically in polynomial time.
\end{enumerate}
All running times are polynomial in the encoding length of the complete payoff table. Their degrees may depend on the fixed parameters.
\end{corollary}

A positive answer to the opening question would imply condition~\textnormal{(i)} with $\eta=0$, and hence condition~\textnormal{(ii)}. Consequently, refuting condition~\textnormal{(ii)} would give a negative answer to the opening question.
In particular, suppose that for some fixed $k_0\ge2$ and $\delta_0\in(0,1)$, no deterministic polynomial-time algorithm finds a $\delta_0$-ANE in arbitrary $k_0$-player games. 
For every fixed $\eps\in(0,1)$ and $\eta\in[0,1)$, our reduction in Theorem~\ref{thm:amplification} then gives a fixed $C$ for which an $(\eps,\eta)$-weak ANE cannot be found in deterministic polynomial time, even in symmetric $C$-player binary-payoff games.

For a quantitative hardness consequence, we use Rubinstein's hardness result for bimatrix games~\cite[Theorem~I.2]{DBLP:conf/focs/Rubinstein16}. We use deterministic PPAD-ETH in that work's convention: solving \textsc{End-of-Line} requires deterministic time $2^{\widetilde\Omega(n)}$, where $n$ is the combined description length of the predecessor and successor circuits.\footnote{Recent preprints by OpenAI~\cite{OpenAIETH26} and Claude (Anthropic)~\cite{AnthropicETH26} claim refutations of ETH through subexponential-time algorithms for 3-SAT. While the reported results achieve polylogarithmic savings in the exponent, they do not directly contradict PPAD-ETH, because its original $2^{\widetilde\Omega(n)}$ formulation~\cite[Conjecture~1]{DBLP:conf/innovations/BabichenkoPR16} already permits such factors.}

\begin{corollary}[Quasipolynomial hardness under PPAD-ETH]\label{cor:ppad-eth}
Assume deterministic PPAD-ETH. For every fixed $\eps\in(0,1)$ and $\eta\in[0,1)$, there is a constant $C=C(\eps,\eta)$ such that finding an $(\eps,\eta)$-weak ANE in a symmetric $C$-player binary-payoff game has no deterministic algorithm running in time
\begin{equation}\label{eq:hardness-class}
 N^{O((\log N)^{1-\xi})}
 \qquad\text{for any fixed }\xi\in(0,1),
\end{equation}
where $N$ is the encoding length of the complete payoff table. 
\end{corollary}

In Section~\ref{sec:complexity}, we give a general running-time transfer result and prove both corollaries. We also show that the recovery calculation in our reduction succeeds from any profile of $H$ where the average maximum unilateral gain is at most $1-(1-\eta)(1-\eps)$, which is a weaker condition than the requirement of an $(\eps,\eta)$-weak ANE.
Finally, we give quantitative bounds on the number of players required for the reduction.

\subsection{Proof overview}\label{subsec:idea}

We prove Theorem~\ref{thm:amplification} by constructing a game $H$ in which every profile has small average payoff, while failure to recover a $\delta$-ANE of $G$ by a recovery calculation from a mixed strategy profile gives every player a deviation with expected payoff close to one.
The contradiction of these two properties with the weak approximation requirement will ensure the success of recovery from a weak ANE of $H$.
Both properties come from a filtering rule applied for $t$ rounds.

A pure action profile of $G$ is called an action tuple.
Each of the $C$ players of $H$ simultaneously chooses a sequence of $t$ action tuples, one for each filtering round, called its current tuple in that round.
In the payoff calculation, each player contributes an initial mass of $1/C$ to its chosen sequence.
Every sequence starts with weight one, and its remaining mass is its initial mass times its current weight.
In each round, the remaining masses determine a distribution of current tuples.
A tuple's score is the average expected payoff of its actions in $G$ against the independent opponent marginals of this distribution.
The filter reduces the total mass to a fixed fraction $q<1$ by decreasing the weights of sequences based on the scores, and a sequence retains its current weight when its current tuple has a sufficiently high score.
A player's binary payoff is one exactly when its sequence's weight is never reduced through all rounds, i.e., its sequence retains the initial mass.
Since the final total mass is $q^t$, the average payoff is at most $q^t$ under every mixed strategy profile.

Given any mixed strategy profile of $H$, the recovery calculation applies the same filtering rule, starting from the uniform mixture of the players' mixed strategies.
In each round, it tests whether most of the probability of the current distribution lies on tuples with nearly maximum score.
If so, the product of the distribution's marginals is a $\delta$-ANE of $G$, because a small score deficit bounds each player's payoff deficit from a best response.
If every test fails, each round has more probability on tuples sufficiently below the maximum score than the filter is permitted to remove.
Choosing a score-maximizing tuple in every round gives a sequence $z^*$ that retains weight one throughout the recovery calculation.

The main analysis aims to show that $z^*$ also gives a high expected payoff in $H$ when a player deviates to it while other players follow the given mixed strategy profile.
As the opponents independently sample their sequences, the payoff calculation starts from random initial masses. But when $C$ is large, the relevant mass, weight, and score quantities in the filtering process remain close in expectation to their counterparts in the recovery calculation starting from the uniform mixture.
To get payoff one, the deviator needs the weight exactly unchanged in every round.
The failed tests supply probabilities on low-scored tuples with enough margins, so that $z^*$ retains weight one unless the probability on the lower-scoring tuples or their score gaps changes substantially.
With other parameters fixed, we control the expected discrepancies in this comparison by $O(C^{-1/2})$, yielding a lower bound of $1-K/\sqrt C$ on the deviation payoff, where $K$ depends only on the fixed parameters and the number of rounds.
The estimates are independent of the action-set sizes because they do not require a simultaneous error bound over all sequences. We only need to control the average weight error over the sampled sequences, the average score error at their current tuples, and the score error at the deviator's current tuple in each round.

Failed recovery thus makes the players' average maximum unilateral gain at least $1-q^t-K/\sqrt C$.
At an $(\eps,\eta)$-weak ANE, this average is at most $(1-\eta)\eps+\eta=1-(1-\eta)(1-\eps)$.
Choosing $t$ and then $C$ sufficiently large makes these bounds incompatible, forcing recovery to succeed.
Both choices are independent of the action-set sizes, so the constructed game's complete payoff table has polynomial size.

\subsection{Related work}\label{subsec:related}

\paragraph{Approximation algorithms and hardness.}
Our reduction preserves the quasipolynomial lower bound obtained by Rubinstein~\cite{DBLP:conf/focs/Rubinstein16} under PPAD-ETH, while extending hardness to every fixed positive approximation parameter below one. For a fixed number of players and a fixed positive approximation parameter, this nearly matches the quasipolynomial-time algorithms obtained from small-support equilibria~\cite{DBLP:conf/sigecom/LiptonMM03,DBLP:journals/mor/BabichenkoBP17}. Golowich~\cite{DBLP:journals/corr/abs-2609-07136} studies the dependence on the approximation parameter when it tends to zero with the input size, proving lower bounds for bimatrix games under both the PCP-for-PPAD and PPAD-ETH conjectures.

Earlier work established PPAD-completeness for approximate Nash equilibrium in games with a fixed number of players, including inverse-polynomial approximation hardness for bimatrix games~\cite{DBLP:journals/siamcomp/DaskalakisGP09,DBLP:journals/jacm/ChenDT09}. Chen, Teng, and Valiant~\cite{DBLP:conf/soda/ChenTV07} show that binary payoffs suffice for this inverse-polynomial hardness. Bil\`o and Mavronicolas~\cite{DBLP:journals/algorithmica/BiloM21} prove PPAD-completeness for finding an exact Nash equilibrium in symmetric binary-payoff bimatrix games.

A line of work has improved the approximation guarantees achievable in polynomial time in games with constant player numbers. For every fixed $\kappa>0$, Tsaknakis and Spirakis~\cite{DBLP:journals/im/TsaknakisS08} obtain a $(0.3393+\kappa)$-ANE in bimatrix games. Deligkas, Fasoulakis, and Markakis~\cite{DBLP:journals/talg/DeligkasFM23} refine their approach to obtain a $(1/3+\kappa)$-ANE, and Li and Li~\cite{LL26} report a further improvement to approximately $0.30954+\kappa$. For three players, Li, Li, and Deng~\cite{LLD26} give a $(1/2+\kappa)$-ANE algorithm. H\'emon, de Rougemont, and Santha~\cite{DBLP:conf/sagt/HemonRS08} show how to transform a polynomial-time algorithm for $\alpha$-ANE in $r$-player games into a polynomial-time algorithm for $1/(2-\alpha)$-ANE in $(r+1)$-player games. This extension carries improvements for bimatrix games to larger fixed numbers of players, although its guarantees still approach one as the number of players grows. Section~\ref{subsec:small} uses the classical $(1-1/p)$ guarantee~\cite{DBLP:journals/corr/abs-0804-4524,DBLP:conf/sagt/HemonRS08} to bound the number of players needed for our hardness result.

\paragraph{Weak approximate Nash equilibrium.}
Babichenko, Papadimitriou, and Rubinstein~\cite{DBLP:conf/innovations/BabichenkoPR16} introduced weak approximate Nash equilibrium and showed its PPAD-hardness in degree-three, two-action polymatrix games under the PCP-for-PPAD conjecture. An OpenAI preprint~\cite{OpenAIPCP26} reports a proof of PCP-for-PPAD with a quasilinear reduction.

For some fixed $\eps,\eta>0$, Rubinstein~\cite{DBLP:conf/focs/Rubinstein16} proves PPAD-hardness for $(\eps,\eta)$-weak ANE in two-action succinct games and an exponential lower bound on the number of payoff queries. Babichenko and Rubinstein~\cite{DBLP:journals/geb/BabichenkoR22} prove an exponential randomized communication lower bound for $(\eps,\eps)$-weak ANE in binary-action games for some constant $\eps>0$. Their construction represents vector choices by populations of players and uses error-correcting codes to tolerate players who violate the approximation requirement.

On the algorithmic side, Arieli and Babichenko~\cite{DBLP:conf/sigecom/ArieliB17} show that, when the number of actions per player and both parameters $\eps,\eta>0$ are fixed, an $(\eps,\eta)$-weak ANE exists on a grid of probabilities whose size is independent of the number of players, and therefore can be found by enumerating profiles within polynomial time in the complete payoff-table length.

\paragraph{Polymatrix and graphical games.}
Polymatrix games and bounded-degree graphical games admit succinct representations when the number of players grows. 
For normalized polymatrix games, where payoffs are sums of pairwise interactions, Deligkas, Fearnley, Savani, and Spirakis~\cite{DBLP:journals/algorithmica/DeligkasFSS17} give a polynomial-time $(1/2+\kappa)$-ANE algorithm for every fixed $\kappa>0$, independently of the number of players. Rubinstein~\cite{DBLP:journals/siamcomp/Rubinstein18} proves PPAD-hardness at a sufficiently small fixed approximation parameter even for degree-three, two-action polymatrix games.
Deligkas, Fearnley, Hollender, and Melissourgos~\cite{DBLP:conf/aaai/DeligkasFHM23,DBLP:journals/jacm/DeligkasFHM24} obtain tight approximation thresholds for two-action graphical games through reductions from their PPAD-complete problem Pure-Circuit. Uniform mixing gives a $1/2$-ANE, and finding an $\eps$-ANE is PPAD-hard for every fixed $\eps<1/2$. For well-supported approximate Nash equilibrium (WSNE), which requires every supported action to be an approximate best response, they prove hardness for every fixed approximation parameter below one, even with two actions and binary payoffs. They conjecture that a polynomial-time $1/2$-ANE algorithm exists for graphical games with arbitrary action sets. Under PPAD-ETH, Corollary~\ref{cor:ppad-eth} rules out such an algorithm at constant degree: our hard games have a constant number of players and can be represented on the complete interaction graph using the same payoff tables.

\paragraph{Related techniques in hardness amplification.}
In our proof, repeated filtering makes the average-payoff bound decrease geometrically, while failure of all recovery tests supplies a deviation with expected payoff close to one for every player, which enables the amplification of the average maximum unilateral gain.
Repetition is a classical tool for hardness amplification.
Goldenberg and Karthik C. S.~\cite{DBLP:conf/innovations/GoldenbergS20} develop a direct-product framework to combine multiple instances to amplify average-case hardness, with an application to approximate Nash equilibrium. 
Chen, Cheng, and Tang~\cite{DBLP:conf/innovations/ChenCT17} use replication and concentration to reduce WSNE to ANE in games with many players and a bounded number of actions. They replace each player by a group, use plurality to determine the group's action, and recover a WSNE by removing actions with small aggregate probability. The number of copies and the loss in the approximation parameter depend on the number of actions.
A similar comparison between average payoff and deviation payoff underlies the support lower bounds of Feder, Nazerzadeh, and Saberi~\cite{DBLP:conf/sigecom/FederNS07} and their multiplayer extensions~\cite{DBLP:journals/corr/abs-0804-4524,DBLP:conf/sagt/HemonRS08}, where the total payoff is bounded so that some player has a low expected payoff, but deviation can obtain payoff one against opponents with sufficiently small supports.

\paragraph{Organization.}
Section~\ref{sec:reduction} defines the filtering construction and recovery calculation and proves Theorem~\ref{thm:amplification} using the deviation estimate established in Section~\ref{sec:deviation}. Section~\ref{sec:complexity} gives the running-time transformation and further consequences, including bounds on the number of players. Omitted proofs appear in the appendices. 

Generative AI tools helped find and simplify the proof of the main result. Their use is described in the AI methodology statement at the end of the paper.

\section{Construction and recovery by multi-round filtering}\label{sec:reduction}

For a $p$-player normal-form game with pure action sets $A_1,\ldots,A_p$, a mixed strategy profile $\sigma=(\sigma_1,\ldots,\sigma_p)$ induces the product distribution $\bigotimes_j\sigma_j$ on pure profiles. Write $U_i(\sigma)$ for player $i$'s expected payoff under this law, and $U_i(a_i,\sigma_{-i})$ for its expected payoff when it chooses $a_i$ and its opponents independently follow $\sigma_{-i}$. An $\eps$-approximate Nash equilibrium ($\eps$-ANE) is a mixed strategy profile $\sigma$ satisfying the following inequality for every player $i$:
\begin{equation}\label{eq:improvement}
 \max_{a_i\in A_i}U_i(a_i,\sigma_{-i})-U_i(\sigma)\le\eps.
\end{equation}
An $(\eps,\eta)$-weak approximate Nash equilibrium, or $(\eps,\eta)$-weak ANE, allows this inequality to fail for at most $\lfloor\eta p\rfloor$ players.

We consider finite normal-form games with payoffs in $[0,1]$. In computational statements, we represent games by their complete payoff tables and encode mixed strategy profiles as explicit lists of action probabilities. Payoffs and probabilities are rational numbers whose numerators and denominators are encoded in binary, and running times are measured in bit operations.

Fix the number $k\ge2$ of players in the input game, the approximation parameter $\delta\in(0,1)$ for that game, the approximation parameter $\eps\in(0,1)$ for the constructed game, and the permitted fraction $\eta\in[0,1)$ of violations.
Our goal is the guarantee in Theorem~\ref{thm:amplification}: given any $k$-player game $G$ of encoding length $M$, the reduction deterministically constructs a symmetric $C$-player game $H$ with binary payoffs in time polynomial in $M$, where $C\ge2$ depends only on $k,\delta,\eps,\eta$. From every $(\eps,\eta)$-weak ANE $\sigma$ of $H$, the recovery calculation computes a $\delta$-ANE of $G$ in time polynomial in $M+L_\sigma$, where $L_\sigma$ is the actual encoding length of the supplied profile. All constants and polynomial degrees may depend on the fixed parameters, but not on the input action-set sizes.

We first define the scores and filtering subroutine, then specify the payoff rule and recovery algorithm, and use their properties to prove Theorem~\ref{thm:amplification}. The proof uses a deviation-payoff estimate stated in Lemma~\ref{lem:deviation} and proved in Section~\ref{sec:deviation}.

\subsection{Scores and the recovery criterion}\label{subsec:scores}

Denote the action sets of $G$ by $A_1,\ldots,A_k$, and write $A=\prod_{r=1}^k A_r$ for its set of pure action tuples. Write $[k]=\{1,\ldots,k\}$. Each player $r\in[k]$ has a payoff function $u_r:A\to[0,1]$. We use $r$ to index players of $G$ and $i,j$ to index players of the constructed game.

Let $Q$ be a probability distribution on $A$. For each $r\in[k]$, let $x_r$ be its marginal on $A_r$, and let $x=(x_1,\ldots,x_k)$ be the mixed strategy profile formed by these marginals. Players choose their actions independently under $x$, even if $Q$ is correlated. We write $u_r(x)$ for player $r$'s expected payoff under $x$, and write $u_r(a_r,x_{-r})$ for player $r$'s expected payoff when it chooses $a_r$ and each opponent $\ell\ne r$ independently uses $x_\ell$.

We now define a test on $Q$ that gives a sufficient condition for $x$ to be an approximate Nash equilibrium. Define the score of each tuple $a\in A$ with respect to $Q$ by
\begin{equation}\label{eq:score}
 V_Q(a)=\frac{1}{k}\sum_{r=1}^k u_r(a_r,x_{-r}).
\end{equation}
The score is the average of the expected payoffs from the actions in $a$, with each player's action evaluated against the product of the opponent marginals of $Q$. It lies in $[0,1]$ and is defined for every tuple $a\in A$, including those outside the support of $Q$.

For fixed $Q$, the $r$-th summand depends only on $a_r$. Thus the maximum score, denoted by $V_Q^*$, is attained by choosing a best response in each coordinate:
\begin{equation}\label{eq:score-max}
 V_Q^*:=\max_{a\in A}V_Q(a)
 =\frac{1}{k}\sum_{r=1}^k
   \max_{b_r\in A_r}u_r(b_r,x_{-r}).
\end{equation}
The next lemma shows that, if $Q$ puts most of its probability on tuples with nearly maximum score, then the product of its marginals is an approximate equilibrium.
\begin{lemma}[Score criterion]\label{lem:criterion}
For any $\beta,\Delta>0$, if
\begin{equation}\label{eq:criterion}
 Q\{a:V_Q^*-V_Q(a)\le\Delta\}\ge1-\beta,
\end{equation}
then $x=(x_1,\ldots,x_k)$ is a $(\beta+k\Delta)$-ANE of $G$.
\end{lemma}

\begin{proof}
Fix a player $r\in[k]$ and keep the opponents' product strategy $x_{-r}$ fixed. Player \(r\)’s expected payoff \(u_r(x)\) can be calculated by averaging \(u_r(a_r,x_{-r})\) over \(a\sim Q\), since $a_r$ has the marginal $x_r$. 
Equations~\eqref{eq:score} and~\eqref{eq:score-max} give, for every $a\in A$,
\[
 0\le\max_{b_r\in A_r}u_r(b_r,x_{-r})-u_r(a_r,x_{-r})
 \le k\bigl(V_Q^*-V_Q(a)\bigr).
\]
The second inequality holds because $k(V_Q^*-V_Q(a))$ is the sum of these nonnegative payoff differences over all players. Each difference is also at most one, since payoffs lie in $[0,1]$. It is therefore at most $k\Delta$ on the event $\{a:V_Q^*-V_Q(a)\le\Delta\}$ and at most one outside that event. We obtain
\[
 \begin{aligned}
 &\max_{b_r\in A_r}u_r(b_r,x_{-r})-u_r(x)\\
 &\quad=\E_{a\sim Q}\left[
       \max_{b_r\in A_r}u_r(b_r,x_{-r})-u_r(a_r,x_{-r})\right]\\
 &\quad\le k\Delta\,Q\{a:V_Q^*-V_Q(a)\le\Delta\}
       +Q\{a:V_Q^*-V_Q(a)>\Delta\}\\
 &\quad\le k\Delta+\beta.
 \end{aligned}
\]
Since $r$ was arbitrary, the claim follows.
\end{proof}

Fix rational constants $\beta,\Delta$ for the score test, depending only on $k,\delta$, and define the filtering parameters $q,h$ as follows:
\begin{equation}\label{eq:source-accuracy}
 0<\beta,\Delta<1,\qquad \beta+k\Delta<\delta,
 \qquad q=1-\frac{\beta}{2},\qquad h=\frac{\Delta}{4}.
\end{equation}
We say that a distribution $Q$ passes the score test when condition~\eqref{eq:criterion} holds. By Lemma~\ref{lem:criterion}, a passing test guarantees that the product of the marginals of $Q$ is a $\delta$-ANE. If the test fails, $Q$ assigns probability greater than $\beta$ to tuples whose scores are more than $\Delta$ below $V_Q^*$.
The same parameters $q,h$ are used in both filtering applications below.
\subsection{One filtering round}\label{subsec:filter}

Fix a sequence length $t\ge1$ and write $z=(z_1,\ldots,z_t)\in A^t$. 
Assign each sequence $z\in A^t$ a fixed initial mass $\rho(z)\ge0$. Filtering maintains a weight $w(z)\in[0,1]$ for each sequence, so its remaining mass is $\rho(z)w(z)$. The initial masses stay fixed while the weights change.

A filtering round takes values $v(z)\in[0,1]$ assigned to sequences and scales each sequence's weight by a multiplier $\phi(v(z),\theta)$, where $\theta$ is a common threshold. Specifically, $\phi(v,\theta)$ is defined by
\begin{equation}\label{eq:phi}
 \phi(v,\theta)=
 \begin{cases}
 0,&v\le\theta,\\
 \frac{v-\theta}{h},&\theta<v<\theta+h,\\
 1,&v\ge\theta+h.
 \end{cases}
\end{equation}
For fixed $\theta$, the function is nondecreasing and $\frac{1}{h}$-Lipschitz in $v$. It takes values in $[0,1]$, and equals one whenever $v\geq\theta+h$. 

Intuitively, the filter favors larger assigned values, and lets a sufficiently high-scoring sequence retain its current weight while reducing the total mass by a fixed fraction.
Formally, we define the following subroutine for a filtering round.
\paragraph{Subroutine $\Filter(\rho,w,v)$.}
The inputs are initial masses $\rho:A^t\to\mathbb R_{\ge0}$, current weights $w:A^t\to[0,1]$, and assigned values $v:A^t\to[0,1]$. The current total mass $\sum_{z\in A^t}\rho(z)w(z)$ must be positive. The parameters $q,h$ are fixed by \eqref{eq:source-accuracy}. The output is an updated weight function $w^+$ and a threshold $\theta$.
\begin{enumerate}[label=\textnormal{\arabic*.},leftmargin=*]
\item Compute $m=\sum_{z\in A^t}\rho(z)w(z)$. Choose the smallest $\theta\in[-h,1]$ satisfying
\begin{equation}\label{eq:filter-target}
 \sum_{z\in A^t}\rho(z)w(z)
       \phi(v(z),\theta)=qm.
\end{equation}
\item For every $z\in A^t$, set
\begin{equation}\label{eq:filter-update}
 w^+(z)=w(z)\phi(v(z),\theta).
\end{equation}
\item Return $(w^+,\theta)$.
\end{enumerate}

The update is defined for every sequence, including those with zero initial mass. The threshold $\theta$ is returned for use in the analysis. 
The payoff and recovery calculations defined in the next two subsections each run $t$ filtering rounds with initial masses summing to one.

The first lemma below shows that the subroutine is well-defined and retains a fraction $q$ of the current total mass. The second gives a sufficient condition for a multiplier to equal one. 

\begin{lemma}[Well-definedness and mass reduction]\label{lem:filter-properties}
Suppose $\rho(z)\ge0$ and $w(z),v(z)\in[0,1]$ for every $z\in A^t$, with $\sum_{z\in A^t}\rho(z)w(z)>0$. Then $\Filter(\rho,w,v)$ is well-defined and returns weights satisfying
\[
 \begin{gathered}
 0\le w^+(z)\le w(z)\le1\quad(\forall\,z\in A^t),\\
 \sum_{z\in A^t}\rho(z)w^+(z)
 =q\sum_{z\in A^t}\rho(z)w(z).
 \end{gathered}
\]
\end{lemma}

\begin{proof}
Let $m=\sum_z\rho(z)w(z)>0$. Regard $\sum_{z\in A^t}\rho(z)w(z)\phi(v(z),\theta)$ on the left-hand side of \eqref{eq:filter-target} as a function of $\theta$. It is continuous and nonincreasing. At $\theta=-h$, every multiplier is one, so its value is $m$. At $\theta=1$, its value is zero. Since $0\leq qm\leq m$, the solution set of \eqref{eq:filter-target} is a nonempty closed set contained in $[-h,1]$ and therefore has a smallest element. 
For the selected threshold $\theta$, $0\le\phi(v(z),\theta)\le1$ and $w^+(z)=w(z)\phi(v(z),\theta)$ imply $0\le w^+(z)\le w(z)\le1$. 
Since \eqref{eq:filter-target} is satisfied, we have $ \sum_z\rho(z)w^+(z)
 =\sum_z\rho(z)w(z)\phi(v(z),\theta)
 =qm=q\sum_z\rho(z)w(z)$.
\end{proof}

\begin{lemma}[Sufficient condition for a multiplier of one]\label{lem:full-multiplier}
Let $\rho(z)\ge0$ and $w(z),v(z)\in[0,1]$ for every $z\in A^t$, with $m=\sum_z\rho(z)w(z)>0$. Let $\theta$ be the threshold returned by $\Filter(\rho,w,v)$. If a value $v_0\in[0,1]$ and a set $L\subseteq A^t$ satisfy
\begin{equation}\label{eq:full-multiplier-conditions}
 \sum_{z\in L}\rho(z)w(z)>(1-q)m,\qquad
 v_0-v(z)>h\quad(\forall\,z\in L),
\end{equation}
then $\phi(v_0,\theta)=1$.
\end{lemma}

\begin{proof}
Set $\theta_0=v_0-h\in[-h,1]$. If $z\in L$, then $v(z)<\theta_0$, so $\phi(v(z),\theta_0)=0$. All other multipliers are at most one. The total remaining mass at threshold $\theta_0$ is consequently at most
\[
 m-\sum_{z\in L}\rho(z)w(z)<m-(1-q)m=qm.
\]
The selected threshold $\theta$ must be smaller than $\theta_0$, since the remaining mass is nonincreasing in the threshold. Thus $v_0>\theta+h$ and $\phi(v_0,\theta)=1$.
\end{proof}

\subsection{The constructed game and its payoff calculation}\label{subsec:game}

Choose the number of rounds $t$ so that $q^t$ is sufficiently small:
\begin{equation}\label{eq:step-count}
 t=\min\left\{s\in\mathbb Z_{\ge1}:q^s\le\frac{(1-\eta)(1-\eps)}{2}\right\}.
\end{equation}
It is a finite positive integer depending only on $k$ and the approximation parameters. The number $C\ge2$ of players will be fixed in Section~\ref{subsec:correctness} using the deviation-payoff estimate.

Write $[C]=\{1,\ldots,C\}$ for the set of player indices in $H$.
Every player of $H$ has pure action set $A^t$. That is, a pure action of a player in $H$ is to choose a sequence of action tuples in $G$. For a pure profile $y=(y_1,\ldots,y_C)$, write $y_j=(y_{j,1},\ldots,y_{j,t})$. The players choose these sequences simultaneously. Their choices define an empirical distribution $\widehat P$ on sequences, whose probabilities serve as the initial masses:
\begin{equation}\label{eq:empirical-initial-masses}
 \widehat P(z)=\frac{1}{C}\sum_{j=1}^C\ind\{y_j=z\}
 \qquad(\forall\,z\in A^t).
\end{equation}
Each player contributes mass $\frac{1}{C}$ to its chosen sequence, so these masses sum to one. We mark quantities computed from these initial masses with hats.

We now define the \emph{payoff calculation}. Initialize $\widehat W_0(z)=1$ for every $z\in A^t$, so the initial total mass is $\sum_{z\in A^t}\widehat P(z)\widehat W_0(z)=1$. 

In each round $s=1,\ldots,t$, the filtering uses the tuple $z_s$ in each sequence $z\in A^t$. Suppose inductively that the total mass entering round \(s\) is \(m_s=q^{s-1}\), as holds initially for \(s=1\). Group the current masses $\widehat P(z)\widehat W_{s-1}(z)$ according to the current tuple $z_s$ and normalize by $m_s$ to define the distribution
\begin{equation}\label{eq:empirical-Q}
 \widehat Q_s(a)=\frac{1}{m_s}\sum_{z\in A^t}
 \widehat P(z)\widehat W_{s-1}(z)\ind\{z_s=a\}
 \qquad(\forall\,a\in A).
\end{equation}
For each sequence $z\in A^t$, use the score of tuple $z_s$ with respect to $\widehat Q_s$ as the assigned value in the update:
\begin{equation}\label{eq:empirical-update}
 (\widehat W_s,\widehat\theta_s)
 =\Filter\bigl(\widehat P,\widehat W_{s-1},
                  z\mapsto V_{\widehat Q_s}(z_s)\bigr).
\end{equation}
The total mass after the update becomes $qm_s=q^s$ by Lemma~\ref{lem:filter-properties}, which by induction gives
\begin{equation}\label{eq:empirical-mean}
 \sum_{z\in A^t}\widehat P(z)\widehat W_s(z)=q^s
 \qquad(\forall\,s\in\{0,\ldots,t\}).
\end{equation}
Therefore, $\widehat Q_s$ is correctly normalized and the update is valid.

After all $t$ rounds, define each player $i$'s payoff by
\begin{equation}\label{eq:binary-payoff}
 U_i(y)=\ind\{\widehat W_t(y_i)=1\}.
\end{equation}
Since the initial weight is one and every multiplier lies in $[0,1]$, a payoff of one means that every round's multiplier at the player's sequence is exactly one. All players choosing the same sequence receive the same payoff.

\begin{lemma}\label{lem:game-properties}
The game $H$ is symmetric with binary payoffs.
\end{lemma}

\begin{proof}
A permutation of players leaves the initial masses in \eqref{eq:empirical-initial-masses} unchanged. Every distribution of current action tuples, score, threshold and weight update is determined by these initial masses, so the same weights are computed after the permutation. The multiplier assigned to a sequence is unchanged in every round. The players' payoffs are therefore permuted with their chosen sequences, proving symmetry. The payoffs are binary by definition.
\end{proof}

\begin{proposition}[Upper bound on average payoff]\label{prop:average-payoff}
Every mixed profile $\sigma$ of $H$ satisfies
\begin{equation}\label{eq:average-payoff}
 \frac{1}{C}\sum_iU_i(\sigma)\le q^t.
\end{equation}
\end{proposition}

\begin{proof}
Fix a pure profile $y$. The weights lie in $[0,1]$, so $\ind\{\widehat W_t(z)=1\}\le\widehat W_t(z)$. Grouping players by their chosen sequences gives
\[
 \frac{1}{C}\sum_iU_i(y)
 =\sum_z\widehat P(z)\ind\{\widehat W_t(z)=1\}
 \le\sum_z\widehat P(z)\widehat W_t(z)=q^t,
\]
where the last equality is \eqref{eq:empirical-mean}.
Taking expectation under any mixed strategy profile proves \eqref{eq:average-payoff}.
\end{proof}

\subsection{The recovery calculation}\label{subsec:recovery}

Given a mixed strategy profile $\sigma=(\sigma_1,\ldots,\sigma_C)$ of $H$, we define the \emph{recovery calculation}, which searches for a $\delta$-ANE of $G$. It applies the filtering rule with the initial masses defined below, computes all $t$ rounds, and tests the distribution of current action tuples from each round. It returns the marginals of the first distribution passing the score test, or reports failure if every score test fails. Section~\ref{subsec:correctness} will show that failure is impossible on an $(\eps,\eta)$-weak ANE for the chosen $C$.

Define $P$ as the uniform mixture of the component mixed strategies:
\begin{equation}\label{eq:mixture}
 P=\frac{1}{C}\sum_j\sigma_j.
\end{equation}
$P$ is a distribution over $A^t$. Notably, $P$ equals the expectation of the empirical distribution $\widehat P$ under $\sigma$. The recovery calculation uses $P(z)$ as the initial mass of sequence $z$ and applies the same filtering process.  

Start with $W_0(z)=1$ for every $z\in A^t$. As in the payoff calculation, assume inductively that the total mass entering round $s$ is $m_s=q^{s-1}>0$. For $s=1,\ldots,t$, define the distribution $Q_s$ of current action tuples and the maximum score by
\begin{equation}\label{eq:Q}
 Q_s(a)=\frac{1}{m_s}\sum_{z\in A^t}P(z)W_{s-1}(z)\ind\{z_s=a\},\qquad
 V_s^*=\max_{a\in A}V_{Q_s}(a).
\end{equation}
Update the weights by
\begin{equation}\label{eq:W}
 (W_s,\theta_s)
 =\Filter\bigl(P,W_{s-1},z\mapsto V_{Q_s}(z_s)\bigr).
\end{equation}
Lemma~\ref{lem:filter-properties} completes the induction, giving
\begin{equation}\label{eq:population-mean}
 \sum_zP(z)W_s(z)=q^s\qquad(\forall\,s\in\{0,\ldots,t\}).
\end{equation}

After all $t$ rounds, apply the score test to each $Q_s$:
\begin{equation}\label{eq:test}
 Q_s\{a:V_s^*-V_{Q_s}(a)\le\Delta\}\ge1-\beta.
\end{equation}
If the score test passes for at least one $Q_s$, choose the smallest such index $s$ and return the mixed strategy profile formed by its marginals. This is a $\delta$-ANE by Lemma~\ref{lem:criterion} and the parameter choice \eqref{eq:source-accuracy}. If none passes, report failure. All operations use exact arithmetic on the initial masses and the weights.

\paragraph{Analysis of failed recovery.}
Suppose the recovery calculation fails under the mixed strategy profile $\sigma$, that is, \eqref{eq:test} fails for every $s$. We show that every player can unilaterally deviate to a strategy $z^*\in A^t$ to obtain a high expected payoff. Specifically, choose $a_s^*\in\argmax_{a\in A}V_{Q_s}(a)$ using a fixed tie-breaking order, and define the sequence  $z^*=(a_1^*,\ldots,a_t^*)$.

For intuition, consider the weight of $z^*$ in the recovery calculation.
In each round $s$, the failed score test means that more than a $\beta$ fraction of the current mass lies on tuples whose scores are more than $\Delta$ below the maximum score $V_s^*=V_{Q_s}(a_s^*)$. Equivalently, \[Q_s\{a:V_{Q_s}(a_s^*)-V_{Q_s}(a)>\Delta\}>\beta.\]
Since $\beta>1-q=\frac{\beta}{2}$ and $\Delta>h=\frac{\Delta}{4}$, Lemma~\ref{lem:full-multiplier} shows that $z^*$ receives multiplier one in every round. Thus $z^*$ keeps weight one throughout the recovery calculation.

In the payoff calculation, each player's sequence contributes only $1/C$ to the initial masses. For any fixed bounded statistic of the sequences, the opponents' independent choices produce an average close to its mean under $P$, while forcing the deviation of a player contributes an additional error of at most $1/C$. The analysis in Section~\ref{sec:deviation} shows that these sampling estimates remain sufficient through the filtering rounds so that $z^*$ retains weight one with probability close to one. The following lemma states the resulting lower bound on the expected payoff from this deviation.

\begin{lemma}[Lower bound on deviation payoff]\label{lem:deviation}
There is an explicitly computable $K=K(k,\beta,\Delta,t)$, independent of $G,A,C$, such that, if the recovery calculation fails on a mixed strategy profile $\sigma$, then for every fixed player $i$,
\begin{equation}\label{eq:deviation}
 U_i(z^*,\sigma_{-i})=\E_{Y_{-i}\sim\sigma_{-i}}U_i(z^*,Y_{-i})
 \ge1-\frac{K}{\sqrt C}.
\end{equation}
Here $z^*$ is the sequence of maximizing tuples from the recovery calculation on the original profile $\sigma$.
\end{lemma}

Section~\ref{sec:deviation} proves this estimate. We first use it to choose $C$ and complete the reduction.

\subsection{Completion of the reduction}\label{subsec:correctness}

With $\beta,\Delta$ fixed by~\eqref{eq:source-accuracy} and $t$ fixed by~\eqref{eq:step-count}, Lemma~\ref{lem:deviation} supplies $K$ depending only on $k,\delta,\eps,\eta$. Choose an integer $C\ge2$ satisfying
\begin{equation}\label{eq:C-choice}
 \sqrt C>\frac{2K}{(1-\eta)(1-\eps)}.
\end{equation}
The construction is now specified.
We now prove that the recovery calculation succeeds at every $(\eps,\eta)$-weak ANE of $H$.

\begin{proof}[Proof of Theorem~\ref{thm:amplification}]
The game $H$ has the required symmetry and binary payoffs by Lemma~\ref{lem:game-properties}. Suppose that the recovery calculation fails at a profile $\sigma$. By Lemma~\ref{lem:deviation}, with the sequence $z^*$ of maximizing tuples from the failed recovery calculation, for every player $i$,
\[
 \max_{a_i\in A^t}U_i(a_i,\sigma_{-i})
 \ge U_i(z^*,\sigma_{-i})
 \ge 1-\frac K{\sqrt C}.
\]
Subtracting $U_i(\sigma)$ and averaging over all players gives
\begin{equation}\label{eq:failed-recovery-average-gain}
 \begin{aligned}
 \frac1C\sum_{i=1}^C
 \left[
   \max_{a_i\in A^t}U_i(a_i,\sigma_{-i})-U_i(\sigma)
 \right]
 &\ge 1-\frac K{\sqrt C}-\frac1C\sum_{i=1}^C U_i(\sigma)\\
 &\ge 1-\frac K{\sqrt C}-q^t\\
 &>1-(1-\eta)(1-\eps).
 \end{aligned}
\end{equation}
The second inequality follows from Proposition~\ref{prop:average-payoff}. The strict inequality uses $K/\sqrt C<(1-\eta)(1-\eps)/2$ from~\eqref{eq:C-choice} and $q^t\le(1-\eta)(1-\eps)/2$ from~\eqref{eq:step-count}.

Now suppose that $\sigma$ is an $(\eps,\eta)$-weak ANE. By definition, at least $C-\lfloor\eta C\rfloor\ge(1-\eta)C$ players have maximum unilateral gain at most $\eps$, and every player's gain is at most one. Hence
\begin{equation}\label{eq:weakANEaveragegain}
 \frac1C\sum_{i=1}^C
 \left[
   \max_{a_i\in A^t}U_i(a_i,\sigma_{-i})-U_i(\sigma)
 \right]
 \le (1-\eta)\eps+\eta
 =1-(1-\eta)(1-\eps).
\end{equation}
This contradicts~\eqref{eq:failed-recovery-average-gain}. Hence recovery cannot fail at an $(\eps,\eta)$-weak ANE: the score test passes for some $Q_s$, and the recovery calculation returns a $\delta$-ANE of $G$.

For the computational guarantees, scores are computed directly from the $k$ marginals and the complete payoff table of $G$, by summing each payoff against the product of the relevant opponent marginals. For fixed $k$, this uses polynomially many rational operations. Appendix~\ref{app:arithmetic} proves exact polynomial-time computation of $\Filter$ and bounds intermediate bit lengths across all rounds.

Each player has $|A|^t$ pure actions, so $H$ has $|A|^{tC}$ pure profiles. The construction enumerates them, evaluates the $t$ filtering rounds, and writes the binary payoffs into the complete table. Since $C,t$ are fixed, the table size and exact computation time are polynomial in the encoding length of $G$.

For the recovery calculation, the supplied probability lists explicitly describe each $\sigma_j$ on $A^t$. The algorithm forms $P$, enumerates sequences to compute the distributions of current action tuples and the weight updates in all $t$ rounds, then checks their score tests. Appendix~\ref{app:arithmetic} bounds the computation and all intermediate bit lengths by a polynomial in the encoding lengths of $G$ and the supplied profile. The returned marginals fit within the same bound. Both algorithms are deterministic, with polynomial degrees depending only on $k$ and the fixed approximation parameters.
\end{proof}

\section{Comparison of the payoff and recovery calculations}\label{sec:deviation}

We prove Lemma~\ref{lem:deviation} by comparing the payoff
calculation after a player deviates to $z^*$ with the recovery calculation on the original profile $\sigma$. Our goal is to show that when the recovery calculation fails, all multipliers at the deviator's sequence $z^*$ equal one with high probability.

Section~\ref{subsec:sampling} defines the deviation experiment and bounds the differences between averages over the sampled sequences and expectations under the original mixture $P$.
Section~\ref{subsec:weight-comparison} shows how differences in the scores and incoming weights affect the updated weights in a single round. While the filtering threshold may move sharply, the proof controls the resulting difference in the updated weights using the monotonicity of the filtering rule. Section~\ref{subsec:errors} combines these estimates to bound the expected score and weight differences over successive rounds. 
Finally, Section~\ref{subsec:payoff-bound} shows that the deviator's sequence is unlikely to lose weight in the payoff calculation. In the recovery calculation, a substantial fraction of the probability lies on tuples scoring far below the deviator's choices. For the deviator to lose weight, these tuples must either lose substantial probability or become much closer in score to the deviator's choices. The preceding estimates show that such changes are unlikely when $C$ is large, yielding the deviation-payoff bound.

\subsection{The deviation experiment and sampling bounds}\label{subsec:sampling}

We first specify the deviation experiment and establish sampling bounds for comparing the payoff and recovery calculations.

Fix a mixed profile $\sigma=(\sigma_1,\ldots,\sigma_C)$ for which the recovery calculation fails, and fix the deviating player $i$. Recall that $A=\prod_{r=1}^k A_r$ is the input game's set of action tuples and that each player of $H$ chooses a sequence in $A^t$. The recovery calculation starts from the fixed mixture $P=\frac1C\sum_j\sigma_j$. Its weight function $W_s:A^t\to[0,1]$ records the weights after round $s$, whereas $Q_s$ is the distribution of current tuples before that round, formed from $W_{s-1}$. Write $\theta_s$ for the threshold used in the recovery calculation in round $s$. All these objects are determined by the original profile.

Use the maximizing tuples $a_s^*\in\argmax_{a\in A}V_{Q_s}(a)$ chosen when the recovery calculation fails, and write $V_s^*=V_{Q_s}(a_s^*)$ and $z^*=(a_1^*,\ldots,a_t^*)$. The parameters remain $q=1-\beta/2$ and $h=\Delta/4$, with $\beta,\Delta\in(0,1)$. Define the deviation experiment by
\begin{equation}\label{eq:experiment}
 Y_i=z^*,\qquad Y_j\sim\sigma_j\ \text{independently for }j\ne i.
\end{equation}
Write $Y_j=(Y_{j,1},\ldots,Y_{j,t})$, so $Y_{j,s}\in A$ is player $j$'s tuple in round $s$. The sequences are drawn once. We write $\E_Y$ and $\Prb_Y$ for expectation and probability under~\eqref{eq:experiment}. These are taken over the opponents' independent draws, with $Y_i=z^*$ fixed. We write $\E_{z\sim P}$ for expectation over a fresh sequence $z$ drawn from the original mixture $P$, which still includes the original strategy $\sigma_i$.

For each realized profile $Y$, the payoff calculation produces the hatted distributions, thresholds and weight functions. In particular, the empirical sequence distribution $\widehat P$ assigns to $z$ the fraction of players choosing $z$, and $\widehat W_s(Y_j)$ is player $j$'s weight after round $s$. The weight function $\widehat W_s$ depends on the whole profile $Y$, whereas the weight function $W_s$ from the recovery calculation is fixed before the draws. The payoff is
\[
 U_i(z^*,Y_{-i})=\ind\{\widehat W_t(Y_i)=1\}.
\]
Both calculations start with weight one and multiply by factors in $[0,1]$. It therefore suffices to show that all of the deviator's multipliers equal one with probability at least $1-K/\sqrt C$.

Each round multiplies total mass by $q$, so both calculations enter round $s$ with mass $m_s=q^{s-1}$. Both calculations have total remaining mass $q^s$ after round $s$, giving the identities
\begin{equation}\label{eq:weight-means}
 \E_{z\sim P}[W_s(z)]=q^s,\qquad
 \frac1C\sum_j\widehat W_s(Y_j)=q^s
 \quad(\forall\,s\in\{0,\ldots,t\}).
\end{equation}
The second identity holds for every realization. The sample average $\frac1C\sum_jW_s(Y_j)$ of the fixed weights from the recovery calculation need not equal $q^s$.

The following lemma bounds the sampling error for a bounded function when some sequences are held fixed.

\begin{lemma}[Sampling bound with fixed sequences]\label{lem:sampling}
Let $I\subseteq[C]$ with $|I|\le\ell$. Fix arbitrary sequences $Y_j$ for $j\in I$ and a function $f:A^t\to[0,1]$. Draw $Y_j\sim\sigma_j$ independently for $j\notin I$. Write $\sigma_{-I}=\bigotimes_{j\notin I}\sigma_j$ for the law of these remaining draws. Then
\begin{equation}\label{eq:sampling}
 \E_{Y_{-I}\sim\sigma_{-I}}\left|\frac{1}{C}\sum_j f(Y_j)-\E_{z\sim P}[f(z)]\right|
 \le\frac{\ell}{C}+\frac{1}{2\sqrt C}.
\end{equation}
The inner expectation is over a fresh draw from $P$.
\end{lemma}
\begin{proof}
Let $\mu_j=\E_{z\sim\sigma_j}[f(z)]$, and define
\[
 S_I=\frac{1}{C}\sum_{j\in I}(f(Y_j)-\mu_j),
 \qquad
 S_{-I}=\frac{1}{C}\sum_{j\notin I}(f(Y_j)-\mu_j).
\]
Since $P=\frac{1}{C}\sum_j\sigma_j$, we have
$\E_{z\sim P}[f(z)]=\frac{1}{C}\sum_j\mu_j$. Thus
\[
 \frac{1}{C}\sum_j f(Y_j)-\E_{z\sim P}[f(z)]
 =S_I+S_{-I}.
\]

The term $S_I$ is fixed. Since $f(Y_j),\mu_j\in[0,1]$,
\[
 |S_I|
 \le\frac{1}{C}\sum_{j\in I}|f(Y_j)-\mu_j|
 \le\frac{|I|}{C}\le\frac{\ell}{C}.
\]

The summands in $S_{-I}$ are independent and centered.
Because $0\le f\le1$, each has variance at most $\frac{1}{4}$.
Consequently,
\[
 \E_{Y_{-I}\sim\sigma_{-I}}[S_{-I}^2]
 =\frac{1}{C^2}\sum_{j\notin I}
       \operatorname{Var}_{z\sim\sigma_j}(f(z))
 \le\frac{C-|I|}{4C^2}\le\frac{1}{4C}.
\]
Cauchy--Schwarz therefore gives
\[
 \E_{Y_{-I}\sim\sigma_{-I}}[|S_{-I}|]
 \le\sqrt{\E_{Y_{-I}\sim\sigma_{-I}}[S_{-I}^2]}
 \le\frac{1}{2\sqrt C}.
\]

The triangle inequality now yields
\[
 \E_{Y_{-I}\sim\sigma_{-I}}
   \left|\frac{1}{C}\sum_j f(Y_j)-\E_{z\sim P}[f(z)]\right|
 \le |S_I|+\E_{Y_{-I}\sim\sigma_{-I}}[|S_{-I}|]
 \le\frac{\ell}{C}+\frac{1}{2\sqrt C}.
\]
\end{proof}

\paragraph{Writing the score as an expectation.}
To apply sampling bounds to scores, we express the score as an expectation of a bounded function of independent tuples. Each player of $G$ has $k-1$ opponents. Given $a,b^{(1)},\ldots,b^{(k-1)}\in A$, evaluate each player $r$'s payoff by using $a_r$ for its own action and taking its opponents' actions from $b^{(1)},\ldots,b^{(k-1)}$, in increasing order of opponent index. Define
\begin{equation}\label{eq:kernel}
 \kappa(a;b^{(1)},\ldots,b^{(k-1)})
 =\frac1k\sum_{r=1}^k
 u_r\bigl(b^{(1)}_1,\ldots,b^{(r-1)}_{r-1},a_r,
             b^{(r)}_{r+1},\ldots,b^{(k-1)}_k\bigr).
\end{equation}
Here the parenthesized superscript labels the input tuple and the subscript labels a player of $G$. The block before $a_r$ is empty when $r=1$, and the block after $a_r$ is empty when $r=k$. The function $\kappa$ takes values in $[0,1]$.

Let $Q$ be a distribution on $A$ with marginals $x_1,\ldots,x_k$, and let $B^{(1)},\ldots,B^{(k-1)}$ be independent tuples, each with law $Q$. In the $r$-th summand of~\eqref{eq:kernel}, the opponents' actions come from distinct independent tuples and have marginals $x_\ell$ for $\ell\ne r$. That summand therefore has expectation $u_r(a_r,x_{-r})$. Averaging over $r$ gives
\begin{equation}\label{eq:score-kernel}
 V_Q(a)=\E_{(B^{(1)},\ldots,B^{(k-1)})\sim Q^{\otimes(k-1)}}
                 [\kappa(a;B^{(1)},\ldots,B^{(k-1)})].
\end{equation}
For each fixed $a$, this identity expresses $V_Q(a)$ as a polynomial of degree at most $k-1$ in the probabilities of $Q$. When $k=2$, the score is linear in $Q$.

The empirical counterpart of~\eqref{eq:score-kernel} averages over ordered tuples of player indices. The next lemma accounts for repeated indices, which reuse one realized sequence, whereas the comparison expectation uses a fresh independent draw in each position.

\begin{lemma}[Sampling bound for functions of several sequences]\label{lem:product-sampling}
Let $I\subseteq[C]$ with $|I|\le\ell$. Fix arbitrary sequences $Y_j$ for $j\in I$ and a function $f:(A^t)^{k-1}\to[0,1]$. Draw $Y_j\sim\sigma_j$ independently for $j\notin I$. For the original mixture $P=\frac1C\sum_j\sigma_j$, the expected absolute difference between the empirical average and the mean of $f$ under $P^{\otimes(k-1)}$ satisfies
\begin{equation}\label{eq:product-sampling}
 \begin{aligned}
 &\E_{Y_{-I}\sim\sigma_{-I}}\Biggl[
   \Biggl|\frac1{C^{k-1}}\sum_{j_1,\ldots,j_{k-1}=1}^C
       f(Y_{j_1},\ldots,Y_{j_{k-1}})\\
 &\qquad\qquad
       -\E_{(Z^{(1)},\ldots,Z^{(k-1)})\sim P^{\otimes(k-1)}}
          [f(Z^{(1)},\ldots,Z^{(k-1)})]\Biggr|\Biggr]\\
 &\quad\le\frac{k}{2\sqrt C}
       +\frac{k\ell+\binom{k}{2}}{C}.
 \end{aligned}
\end{equation}
The outer expectation is over the independent sequences indexed outside $I$. The inner expectation uses $k-1$ fresh independent sequences, each with law $P$. The order of the arguments of $f$ may matter.
\end{lemma}
The proof of Lemma~\ref{lem:product-sampling} is given in Appendix~\ref{app:product-sampling}.

\subsection{Comparing the weight updates}\label{subsec:weight-comparison}

We next show how differences in the incoming weights and scores affect the weights after one update. The comparison holds for every realized profile.

Define the average weight error $e_s^{W}$ after round $s$, the average score error $e_s^{V}$ at the sampled tuples in that round, and the score error $e_s^{*}$ at the fixed maximizing tuple $a_s^*$ by
\begin{align}
 e_s^{W}&=\frac1C\sum_j|\widehat W_s(Y_j)-W_s(Y_j)|
       \quad(\forall\,s\in\{0,\ldots,t\}),\label{eq:r}\\
 e_s^{V}&=\frac1C\sum_j|V_{\widehat Q_s}(Y_{j,s})-V_{Q_s}(Y_{j,s})|
       \quad(\forall\,s\in\{1,\ldots,t\}),\label{eq:e}\\
 e_s^{*}&=|V_{\widehat Q_s}(a_s^*)-V_{Q_s}(a_s^*)|
       =|V_{\widehat Q_s}(a_s^*)-V_s^*|
       \quad(\forall\,s\in\{1,\ldots,t\}).\label{eq:zeta}
\end{align}
The averages in $e_s^{W}$ and $e_s^{V}$ are unweighted averages over player indices, and $e_0^{W}=0$.

\begin{lemma}\label{lem:comparison}
For every round $1\le s\le t$ and every realization of \eqref{eq:experiment},
\begin{equation}\label{eq:weight-difference}
 e_s^{W}\le2e_{s-1}^{W}+\frac2h e_s^{V}+
       \left|\frac1C\sum_jW_s(Y_j)-q^s\right|.
\end{equation}
\end{lemma}
\begin{proof}
Fix a round $s$ and a realized profile $Y$. Define
\[
 \widetilde W_s(z)
 =W_{s-1}(z)\phi(V_{Q_s}(z_s),\widehat\theta_s)
 \qquad(\forall\,z\in A^t).
\]
This intermediate update uses the incoming weights and scores
from the recovery calculation, together with the threshold
from the payoff calculation.

\paragraph{Comparing updates at the same threshold.}
We first compare $\widehat W_s$ with $\widetilde W_s$.
For each player $j$, use the update formulas and add and subtract
$W_{s-1}(Y_j)\phi(V_{\widehat Q_s}(Y_{j,s}),\widehat\theta_s)$.
The triangle inequality gives
\[
 \begin{aligned}
 |\widehat W_s(Y_j)-\widetilde W_s(Y_j)|
 &\le
 |\widehat W_{s-1}(Y_j)-W_{s-1}(Y_j)|
 \phi(V_{\widehat Q_s}(Y_{j,s}),\widehat\theta_s)\\
 &\quad+
 W_{s-1}(Y_j)
 \bigl|\phi(V_{\widehat Q_s}(Y_{j,s}),\widehat\theta_s)
       -\phi(V_{Q_s}(Y_{j,s}),\widehat\theta_s)\bigr|\\
 &\le
 |\widehat W_{s-1}(Y_j)-W_{s-1}(Y_j)|
 +\frac1h|V_{\widehat Q_s}(Y_{j,s})-V_{Q_s}(Y_{j,s})|.
 \end{aligned}
\]
The second inequality uses $0\le\phi\le1$,  $W_{s-1}(Y_j)\le1$ and the
$1/h$-Lipschitz property of $\phi$ on its score argument.
Averaging over $j$ therefore gives
\begin{equation}\label{eq:intermediate-weight-error}
 \frac1C\sum_j|\widehat W_s(Y_j)-\widetilde W_s(Y_j)|\leq e_{s-1}^{W}+\frac{e_s^{V}}h.
\end{equation}

\paragraph{Comparing the two thresholds.}
The recovery update satisfies
\[
 W_s(z)=W_{s-1}(z)\phi(V_{Q_s}(z_s),\theta_s)
 \qquad(\forall\,z\in A^t).
\]
Thus $\widetilde W_s$ and $W_s$ differ only in their thresholds.
Since $\phi$ is nonincreasing in the threshold and
$W_{s-1}(z)\ge0$, we have that
$\widetilde W_s(z)\ge W_s(z)$ holds for every $z$ when $\widehat\theta_s\le\theta_s$, and 
$\widetilde W_s(z)\le W_s(z)$ holds for every $z$ when $\widehat\theta_s\ge\theta_s$.
Consequently, the differences
$\widetilde W_s(Y_j)-W_s(Y_j)$ are either all nonnegative
or all nonpositive.

Their average absolute value therefore equals the absolute
value of their average. We have
\begin{equation}\label{eq:average-abs-samesign}
 \begin{aligned}
 \frac1C\sum_j|\widetilde W_s(Y_j)-W_s(Y_j)|
 &=\left|\frac1C\sum_j\widetilde W_s(Y_j)
          -\frac1C\sum_jW_s(Y_j)\right|\\
 &=\left|
       \frac1C\sum_j\bigl(\widetilde W_s(Y_j)-\widehat W_s(Y_j)\bigr)
       +q^s-\frac1C\sum_jW_s(Y_j)\right|\\
 &\le
 \left|\frac1C\sum_j
       \bigl(\widetilde W_s(Y_j)-\widehat W_s(Y_j)\bigr)\right|
 +\left|q^s-\frac1C\sum_jW_s(Y_j)\right|\\
 &\le
 \frac1C\sum_j|\widetilde W_s(Y_j)-\widehat W_s(Y_j)|
 +\left|\frac1C\sum_jW_s(Y_j)-q^s\right|.
 \end{aligned}
\end{equation}
The second equality is by inserting the payoff calculation's average weight $\frac1C\sum_j\widehat W_s(Y_j)=q^s$.
The first and second inequalities use the triangle inequality.

\paragraph{Combining the comparisons.}
By the triangle inequality, we have
 \begin{align*}
 e_s^{W}
 &=\frac1C\sum_j|\widehat W_s(Y_j)-W_s(Y_j)|\\
 &\le
 \frac1C\sum_j|\widehat W_s(Y_j)-\widetilde W_s(Y_j)|
 +\frac1C\sum_j|\widetilde W_s(Y_j)-W_s(Y_j)|\\
 &\le
 \frac2C\sum_j|\widehat W_s(Y_j)-\widetilde W_s(Y_j)|
 +\left|\frac1C\sum_jW_s(Y_j)-q^s\right|\\
 &\le
 2e_{s-1}^{W}+\frac2h e_s^{V}
 +\left|\frac1C\sum_jW_s(Y_j)-q^s\right|.
 \end{align*}
The second inequality uses~\eqref{eq:average-abs-samesign}.
The last inequality uses~\eqref{eq:intermediate-weight-error}.
\end{proof}

\subsection{Expected errors across rounds}\label{subsec:errors}

We now combine the sampling bounds with the weight comparison to control the expected errors over successive rounds.

\begin{lemma}[Bounds on expected errors]\label{lem:errors}
Put $\Gamma_k=3k^2$ and $E_0^{W}=0$. For $s=1,\ldots,t$, define deterministic coefficients by
\begin{equation}\label{eq:main-error-coefficients}
 E_s^{V}=\frac{(k-1)E_{s-1}^{W}+\Gamma_k}{m_s^{k-1}},\qquad
 E_s^{W}=2E_{s-1}^{W}+\frac{2E_s^{V}}{h}+2.
\end{equation}
Then, for every $1\le s\le t$,
\begin{equation}\label{eq:main-error-bounds}
 \max\{\E_Y[e_s^{V}],\E_Y[e_s^{*}]\}\le\frac{E_s^{V}}{\sqrt C},
 \qquad \E_Y[e_s^{W}]\le\frac{E_s^{W}}{\sqrt C}.
\end{equation}
These bounds are uniform in $G,A$, the original profile and the deviator. The coefficients are computable from $k,\beta,\Delta$ and the round index. For fixed $k,\beta,\Delta,t$, all three expectations are $O_{k,\beta,\Delta,t}(1/\sqrt C)$.
\end{lemma}

\begin{proof}
Fix a round $s$. We first bound the score errors using the weight
error before the round, then use the score bounds to control the
weight error after the update.

\paragraph{Separating weight error from sampling error.}
To apply Lemma~\ref{lem:product-sampling}, we need a fixed function.
For each fixed tuple $a\in A$, define
\begin{equation}\label{eq:fixed-score-kernel}
 f_a(z^{(1)},\ldots,z^{(k-1)})
 =\left(\prod_{d=1}^{k-1}W_{s-1}(z^{(d)})\right)
       \kappa(a;z^{(1)}_s,\ldots,z^{(k-1)}_s)\in[0,1].
\end{equation}
This function uses the weights from the recovery calculation, so it
is determined by the original profile $\sigma$ before any sequences
are sampled.
Expanding~\eqref{eq:score-kernel} using the definitions of
$\widehat Q_s$ and $Q_s$ and multiplying by $m_s^{k-1}$ to cancel the normalizing factors, we have
\begin{equation}\label{eq:weighted-scores}
 \begin{aligned}
 m_s^{k-1}V_{\widehat Q_s}(a)
 &=\frac1{C^{k-1}}\sum_{j_1,\ldots,j_{k-1}=1}^C
       \left(\prod_{d=1}^{k-1}\widehat W_{s-1}(Y_{j_d})\right)
       \kappa(a;Y_{j_1,s},\ldots,Y_{j_{k-1},s}),\\
 m_s^{k-1}V_{Q_s}(a)
 &=\E_{(Z^{(1)},\ldots,Z^{(k-1)})\sim P^{\otimes(k-1)}}
       [f_a(Z^{(1)},\ldots,Z^{(k-1)})].
 \end{aligned}
\end{equation}
We compare these two expressions through the sample average
\[
 S_a=\frac1{C^{k-1}}\sum_{j_1,\ldots,j_{k-1}=1}^C
       f_a(Y_{j_1},\ldots,Y_{j_{k-1}}).
\]
The sum includes all ordered index tuples, including repeated
indices. The only difference between $S_a$ and the first line
of~\eqref{eq:weighted-scores} is that $S_a$ uses $W_{s-1}$ in place
of $\widehat W_{s-1}$. 
It follows that
 \begin{align*}
  |m_s^{k-1}V_{\widehat Q_s}(a)-S_a|&\leq\frac1{C^{k-1}}\sum_{j_1,\ldots,j_{k-1}=1}^C
   \left|\prod_{d=1}^{k-1}\widehat W_{s-1}(Y_{j_d})
             -\prod_{d=1}^{k-1} W_{s-1}(Y_{j_d})\right|
       \kappa(a;Y_{j_1,s},\ldots,Y_{j_{k-1},s})\\
   &\leq\frac1{C^{k-1}}\sum_{j_1,\ldots,j_{k-1}=1}^C
\left|\prod_{d=1}^{k-1}\widehat W_{s-1}(Y_{j_d})
         -\prod_{d=1}^{k-1} W_{s-1}(Y_{j_d})\right|.
 \end{align*}
For any $p_d,q_d\in[0,1]$, $d=1,\ldots,k-1$, changing one factor at a
time gives
\[ \left|\prod_{d=1}^{k-1}p_d-\prod_{d=1}^{k-1}q_d\right|
 \le\sum_{d=1}^{k-1}|p_d-q_d|.\]
Applying this bound, we have
\begin{equation}\label{eq:weight-replacement}
 \begin{aligned}
  |m_s^{k-1}V_{\widehat Q_s}(a)-S_a|
 &\leq\frac1{C^{k-1}}\sum_{j_1,\ldots,j_{k-1}=1}^C\sum_{d=1}^{k-1}
       |\widehat W_{s-1}(Y_{j_d})-W_{s-1}(Y_{j_d})|\\
   &=(k-1)\frac1C\sum_{j=1}^{C}
       |\widehat W_{s-1}(Y_{j})-W_{s-1}(Y_{j})|\\
 &=(k-1)\,e_{s-1}^{W}.
 \end{aligned}
\end{equation}
The first equality holds because each player index occurs in each
position $d$ in exactly $C^{k-2}$ ordered tuples. The second equality
uses the definition of $e_{s-1}^{W}$.
The triangle inequality now decomposes the score difference into
the change in weights and the sampling error for the fixed
function $f_a$:
\begin{equation}\label{eq:score-difference}
 m_s^{k-1}|V_{\widehat Q_s}(a)-V_{Q_s}(a)|
 \le(k-1)e_{s-1}^{W}
       +|S_a- m_s^{k-1}V_{Q_s}(a)|.
\end{equation}

\paragraph{Sampling at the tuples whose scores are needed.}
To bound $e_s^V$ and $e_s^*$, we estimate the second term
in~\eqref{eq:score-difference} at $a=Y_{j,s}$ for each player $j$.
Fix a player index $j$ and put $I=\{i,j\}$. Hold $Y_i=z^*$ and $Y_j=y$ fixed, with $y=z^*$ when $j=i$, and draw $Y_{-I}\sim\sigma_{-I}$. For each such $y$, the function $f_{y_s}$ is fixed. Since $|I|\le2$, Lemma~\ref{lem:product-sampling} gives
\begin{equation}\label{eq:product-sampling-coarse}
 \E_{Y_{-I}\sim\sigma_{-I}}\bigl[|S_{y_s}-m_s^{k-1}V_{Q_s}(y_s)|\bigr]
 \le\frac{k}{2\sqrt C}
       +\frac{2k+\binom{k}{2}}{C}
 \le\frac{\Gamma_k}{\sqrt C}.
\end{equation}
The expectation on the left holds $y$ and $z^*$ fixed. The last inequality uses $C\ge1$, $k\ge2$, and $\Gamma_k=3k^2$.
For $j\ne i$, average this estimate over $y\sim\sigma_j$ to recover the deviation experiment~\eqref{eq:experiment}. For $j=i$, the fixed value $y=z^*$ already gives that experiment. Taking expectation in~\eqref{eq:score-difference} and using the sampling bound gives
\[
\E_Y\bigl[m_s^{k-1}|V_{\widehat Q_s}(Y_{j,s})-V_{Q_s}(Y_{j,s})|\bigr]
 \le(k-1)\E_Y\left[e_{s-1}^{W}\right]+\frac{\Gamma_k}{\sqrt C}.\] 
Averaging over $j$ bounds
$\E_Y[e_s^{V}]$, and taking $j=i$ bounds $\E_Y[e_s^{*}]$, since
$Y_i=z^*$ and $Y_{i,s}=a_s^*$. Hence
\begin{equation}\label{eq:score-expectations}
 \max\{\E_Y[e_s^{V}],\E_Y[e_s^{*}]\}
 \le\frac{(k-1)\E_Y[e_{s-1}^{W}]+\Gamma_k/\sqrt C}{m_s^{k-1}}.
\end{equation}
This argument bounds the average error at the sampled tuples and
the error at the deviator's tuple without a simultaneous sampling
bound over $A$.

\paragraph{Bounding the weight error after the update.}
By Lemma~\ref{lem:comparison}, the updated weight error satisfies
\begin{equation}\label{eq:weight1}
    e_s^{W}\le2e_{s-1}^{W}+\frac2h e_s^{V}+
       \left|\frac1C\sum_jW_s(Y_j)-q^s\right|.
\end{equation}
The function $W_s$ is fixed and has mean $q^s$ under $P$.
Applying Lemma~\ref{lem:sampling} with $I=\{i\}$ and $\ell=1$ gives
\[
 \E_Y\left[\left|\frac1C\sum_jW_s(Y_j)-q^s\right|\right]
 \le\frac1C+\frac1{2\sqrt C}\le\frac2{\sqrt C}.
\]
Taking expectation in \eqref{eq:weight1} therefore gives
\begin{equation}\label{eq:weight-expectations}
 \E_Y[e_s^{W}]
 \le2\E_Y[e_{s-1}^{W}]+\frac2h\E_Y[e_s^{V}]+\frac2{\sqrt C}.
\end{equation}

We finish by induction using~\eqref{eq:main-error-coefficients}.
Initially $e_0^{W}=E_0^{W}=0$.
Assume $\E_Y[e_{s-1}^{W}]\le E_{s-1}^{W}/\sqrt C$.
Substituting into~\eqref{eq:score-expectations} and using the
definition of $E_s^{V}$ gives
\[
 \max\{\E_Y[e_s^{V}],\E_Y[e_s^{*}]\}
 \le\frac{(k-1)E_{s-1}^{W}+\Gamma_k}{m_s^{k-1}\sqrt C}
 =\frac{E_s^{V}}{\sqrt C}.
\]
Substituting this score bound and the induction
hypothesis into~\eqref{eq:weight-expectations} yields
\[
 \E_Y[e_s^{W}]
 \le\frac{2E_{s-1}^{W}+2E_s^{V}/h+2}{\sqrt C}
 =\frac{E_s^{W}}{\sqrt C}.\qedhere
\]
\end{proof}

\subsection{From expected errors to deviation payoff}\label{subsec:payoff-bound}
We show that the score gaps keeping the deviator's weight at one in the failed recovery calculation also protect its weight in the payoff calculation, except when the distributions or scores differ substantially. The preceding estimates show that this is unlikely when $C$ is large.

\begin{proof}[Proof of Lemma~\ref{lem:deviation}]
We first bound the probability that the deviator's multiplier is
below one in a fixed round $s$. Recall that $m_s=q^{s-1}$ is the
total mass entering the round, and define
\[
 D_s=\{a\in A:V_s^*-V_{Q_s}(a)>\Delta\}.
\]
The failed score test for $Q_s$ gives $Q_s(D_s)>\beta$.

\paragraph{Comparing the probabilities assigned to $D_s$.}
In the formula for $\widehat Q_s(D_s)$, replace
$\widehat W_{s-1}$ by $W_{s-1}$, and use it to compare $\widehat Q_s(D_s)$ and $Q_s(D_s)$. By the triangle inequality, we have
\[
 \begin{aligned}
 |\widehat Q_s(D_s)-Q_s(D_s)|
 &\le\frac1{Cm_s}\sum_j
       |\widehat W_{s-1}(Y_j)-W_{s-1}(Y_j)|
       \ind\{Y_{j,s}\in D_s\}\\
 &\quad+\frac1{m_s}\left|
       \frac1C\sum_j W_{s-1}(Y_j)\ind\{Y_{j,s}\in D_s\}
       -m_sQ_s(D_s)\right|\\
 &\le\frac{e_{s-1}^{W}}{m_s}
       +\frac1{m_s}\left|
       \frac1C\sum_j W_{s-1}(Y_j)\ind\{Y_{j,s}\in D_s\}
       -m_sQ_s(D_s)\right|.
 \end{aligned}
\]
The second inequality drops the indicator from the nonnegative summands and uses the definition of $e_{s-1}^{W}$.

Both $D_s$ and $W_{s-1}$ are determined by the original profile.
The function $z\mapsto W_{s-1}(z)\ind\{z_s\in D_s\}$ is therefore fixed, and takes values in $[0,1]$. Its mean under $P$ is $m_sQ_s(D_s)$.
Taking expectation over $Y_{-i}\sim\sigma_{-i}$ and applying
Lemma~\ref{lem:sampling} with $I=\{i\}$ and $\ell=1$ gives the first inequality below:
\begin{equation}\label{eq:mass-discrepancy}
 \E_Y\bigl[|\widehat Q_s(D_s)-Q_s(D_s)|\bigr]
 \le\frac{\E_Y[e_{s-1}^{W}]}{m_s}
       +\frac1{m_s}\left(\frac1C+\frac1{2\sqrt C}\right)
 \le\frac{E_{s-1}^{W}+2}{m_s\sqrt C}.
\end{equation}
For the second inequality, use
$\E_Y[e_{s-1}^{W}]\le E_{s-1}^{W}/\sqrt C$
from Lemma~\ref{lem:errors} for $s\ge2$, together with
$1/C+1/(2\sqrt C)\le2/\sqrt C$.
For $s=1$, use $e_0^{W}=E_0^{W}=0$.

\paragraph{Bounding $\widehat Q_s(D_s)$ when the multiplier is below one.}
Define the set of tuples whose scores in the payoff calculation
are more than $h$ below the score of $a_s^*$:
\[
 B_s=\{a\in A:
       V_{\widehat Q_s}(a_s^*)-V_{\widehat Q_s}(a)>h\}.
\]
Apply Lemma~\ref{lem:full-multiplier} to the filtering call with
initial masses $\widehat P$, incoming weights $\widehat W_{s-1}$,
assigned values $z\mapsto V_{\widehat Q_s}(z_s)$, comparison value
$V_{\widehat Q_s}(a_s^*)$, and set $\{z\in A^t:z_s\in B_s\}$.
The remaining mass of this set is $m_s\widehat Q_s(B_s)$.
Thus, if $\widehat Q_s(B_s)>1-q$, the lemma gives multiplier one
at $a_s^*$. Equivalently, if
$\phi(V_{\widehat Q_s}(a_s^*),\widehat\theta_s)<1$, then
\begin{equation}\label{eq:multiplier-failure-mass}
 \widehat Q_s(B_s)\le1-q=\frac\beta2.
\end{equation}

We next bound $\widehat Q_s(D_s\setminus B_s)$ using the score
errors. For $a\in D_s\setminus B_s$, the definition of $B_s$ gives
$V_{\widehat Q_s}(a_s^*)-V_{\widehat Q_s}(a)\le h$.
It holds that
\[
 \Delta<V_s^*-V_{Q_s}(a)
 \leq |V_s^*-V_{\widehat Q_s}(a_s^*)|+(V_{\widehat Q_s}(a_s^*)-V_{\widehat Q_s}(a))+|V_{\widehat Q_s}(a)-V_{Q_s}(a)|
 \le e_s^{*}+h+|V_{\widehat Q_s}(a)-V_{Q_s}(a)|.
\]
The first inequality uses $a\in D_s$. The second is by the triangle inequality. The last inequality uses the bound on $V_{\widehat Q_s}(a_s^*)-V_{\widehat Q_s}(a)$ and the definition
$e_s^{*}=|V_{\widehat Q_s}(a_s^*)-V_s^*|$.

Summing both sides of $\Delta-h\leq|V_{\widehat Q_s}(a)-V_{Q_s}(a)|+ e_s^{*}$  over $a\in D_s\setminus B_s$ with weights $\widehat Q_s(a)$ gives
\begin{equation}\label{eq:small-gap-mass}
 \begin{aligned}
 (\Delta-h)\widehat Q_s(D_s\setminus B_s)
 &\le\sum_{a\in D_s\setminus B_s}\widehat Q_s(a)|V_{\widehat Q_s}(a)-V_{Q_s}(a)|+e_s^{*}\widehat Q_s(D_s\setminus B_s)\\
 &\le\sum_{a\in A}\widehat Q_s(a)|V_{\widehat Q_s}(a)-V_{Q_s}(a)|+e_s^{*}\\
 &=\frac1{Cm_s}\sum_j\widehat W_{s-1}(Y_j)
       |V_{\widehat Q_s}(Y_{j,s})-V_{Q_s}(Y_{j,s})|+e_s^{*}\\
 &\le\frac{e_s^{V}}{m_s}+e_s^{*}.
 \end{aligned}
\end{equation}
The equality expands $\widehat Q_s$ as the weighted distribution
of the sampled current tuples. The last inequality uses
$\widehat W_{s-1}\le1$ and the definition of $e_s^{V}$.

Suppose now that the deviator's multiplier in round $s$ is below
one. Split $D_s$ into $D_s\cap B_s$ and $D_s\setminus B_s$.
Since $\widehat Q_s(D_s\cap B_s)\le\widehat Q_s(B_s)$,
equations~\eqref{eq:multiplier-failure-mass}
and~\eqref{eq:small-gap-mass} give
\[\widehat Q_s(D_s)=\widehat Q_s(D_s\cap B_s)+\widehat Q_s(D_s\setminus B_s)
\leq\frac\beta2 +\frac{e_s^{V}/m_s+e_s^{*}}{\Delta-h}.
\]
Recall that $Q_s(D_s)>\beta$ by the failed score test. 
Consequently, a multiplier below one implies
\[
 |\widehat Q_s(D_s)-Q_s(D_s)|>\frac\beta2-\frac{e_s^{V}/m_s+e_s^{*}}{\Delta-h}.
\]

\paragraph{Bounding the probability of a multiplier below one.}
Rewrite
\[
 |\widehat Q_s(D_s)-Q_s(D_s)|+\frac{e_s^{V}/m_s+e_s^{*}}{\Delta-h}>\frac\beta2.
\]
By Markov's inequality, we obtain
\begin{equation}\label{eq:round-failure}
 \begin{aligned}
 \Prb_Y\bigl[\phi(V_{\widehat Q_s}(a_s^*),\widehat\theta_s)<1\bigr]&\leq\Prb_Y\left[|\widehat Q_s(D_s)-Q_s(D_s)|+\frac{e_s^{V}/m_s+e_s^{*}}{\Delta-h}>\frac\beta2\right]\\
 &\le\frac2\beta\left[
       \E_Y\bigl[|\widehat Q_s(D_s)-Q_s(D_s)|\bigr]
       +\frac{\E_Y[e_s^{V}]/m_s+\E_Y[e_s^{*}]}{\Delta-h}\right]\\
 &\le\frac{2}{\beta\sqrt C}\left[
       \frac{E_{s-1}^{W}+2}{m_s}
       +\frac{E_s^{V}}{3h}\left(\frac1{m_s}+1\right)\right].
 \end{aligned}
\end{equation}
The last inequality uses~\eqref{eq:mass-discrepancy} for the
expected mass discrepancy, Lemma~\ref{lem:errors} for both
$\E_Y[e_s^{V}]$ and $\E_Y[e_s^{*}]$, and $\Delta-h=3h$.

\paragraph{Summing over the rounds.}
Define
\begin{equation}\label{eq:K}
 K=\sum_{s=1}^t\left[
       \frac{4(E_{s-1}^{W}+2+4E_s^{V}/\Delta)}{\beta q^{s-1}}
       +\frac{4E_s^{V}}{\Delta}\right].
\end{equation}
By~\eqref{eq:K-round-bound} in Appendix~\ref{app:constants}, the
coefficient of $1/\sqrt C$ in~\eqref{eq:round-failure} is at most
the $s$-th summand of $K$. The recurrences
in~\eqref{eq:main-error-coefficients} make $K$ computable from
$k,\beta,\Delta,t$ alone.

The deviator starts with weight one, and every multiplier lies
in $[0,1]$. Its payoff is therefore one exactly when all its
multipliers equal one. The union bound over the rounds gives
\[
 \begin{aligned}
 \E_{Y_{-i}\sim\sigma_{-i}}[U_i(z^*,Y_{-i})]
 &=1-\Prb_Y\left(\bigcup_{s=1}^t
       \{\phi(V_{\widehat Q_s}(a_s^*),\widehat\theta_s)<1\}\right)\\
 &\ge1-\sum_{s=1}^t
       \Prb_Y\bigl(\phi(V_{\widehat Q_s}(a_s^*),\widehat\theta_s)<1\bigr)\\
 &\ge1-\frac K{\sqrt C}.
 \end{aligned}
\]
\end{proof}

\section{Running-time transfer and further consequences}\label{sec:complexity}\label{sec:players}

We first transfer running-time bounds through the construction and recovery algorithms, then prove Corollaries~\ref{cor:polytime} and~\ref{cor:ppad-eth}. We then discuss the average-gain guarantee in the proof of Theorem~\ref{thm:amplification} and compare sufficient and necessary bounds on the number of players.

\subsection{Transfer of running times}\label{subsec:time-transfer}

Write $M$ for the encoding length of the input game, $N$ for that of the constructed game, and $L_\sigma$ for the length of a supplied profile. Running-time bounds are nondecreasing and apply to inputs of length at most their argument. We extend these bounds to positive real arguments using the same convention.
\begin{theorem}[Running-time transfer]\label{thm:hardness-amplification}
Fix an integer $k\ge2$, $\delta,\eps\in(0,1)$ and $\eta\in[0,1)$. There are integers $C\ge2$ and $a_0,b_0\ge1$, depending only on these parameters, with the following guarantees.

Write $M$ for the encoding length of the input game and $N$ for that of the constructed game. If a deterministic algorithm finds an $(\eps,\eta)$-weak ANE in symmetric $C$-player binary-payoff games in time $S(N)$, then a deterministic algorithm finds a $\delta$-ANE in $k$-player games with payoffs in $[0,1]$ in time
\begin{equation}\label{eq:time-transfer}
 O\!\left((M+S(M^{b_0}))^{a_0}\right).
\end{equation}

Furthermore, let $T$ be nondecreasing. Suppose no deterministic algorithm for the input problem runs in time $o(T(M))$, and
\begin{equation}\label{eq:overhead-condition}
 M^{a_0}=o(T(M)).
\end{equation}
Then no deterministic algorithm for the weak-ANE problem runs in time $o(T'(N))$, where
\begin{equation}\label{eq:time-lower-bound}
 T'(N)=\left[T\!\left(N^{1/b_0}\right)\right]^{1/a_0}.
\end{equation}

Both conclusions also hold when the weak-ANE algorithm is guaranteed only on the constructed games. 
\end{theorem}

\begin{proof}
By Theorem~\ref{thm:amplification}, choose integers $a_0,b_0\ge1$ such that construction takes $O(M^{a_0})$ time, the constructed length satisfies $N\le M^{b_0}$, and recovery takes $O((M+L_\sigma)^{a_0})$ time.

The weak-ANE algorithm takes at most $S(M^{b_0})$ time on the constructed game. Since it writes the profile explicitly, its output length satisfies
\[
 L_\sigma=O(S(M^{b_0})).
\]
Thus the total running time is
\[
 \begin{aligned}
 &O\!\left(M^{a_0}+S(M^{b_0})
             +(M+L_\sigma)^{a_0}\right)\\
 &\qquad=O\!\left((M+S(M^{b_0}))^{a_0}\right),
 \end{aligned}
\]
proving~\eqref{eq:time-transfer}.

For the lower-bound implication, suppose a weak-ANE algorithm has running time $S(N)=o(T'(N))$. Since
\[
 T'(M^{b_0})=T(M)^{1/a_0},
\]
we obtain
\[
 S(M^{b_0})^{a_0}=o(T(M)).
\]
Together with~\eqref{eq:overhead-condition}, this gives
\[
 \begin{aligned}
 (M+S(M^{b_0}))^{a_0}
 &\le 2^{a_0-1}
       \left(M^{a_0}+S(M^{b_0})^{a_0}\right)\\
 &=o(T(M)).
 \end{aligned}
\]
The composed algorithm would therefore contradict the assumed lower bound for the input problem.
\end{proof}

\subsection{Polynomial-time algorithms at fixed parameters}\label{subsec:polytime}

\begin{proof}[Proof of Corollary~\ref{cor:polytime}]
Assume condition~\textnormal{(i)} of the corollary, and let $\eps,\eta$ be the fixed parameters it supplies. Fix $k\ge2$ and $\delta\in(0,1)$. Theorem~\ref{thm:hardness-amplification} chooses constants $C,a_0,b_0$ for these parameters. Condition~\textnormal{(i)} supplies a polynomial-time weak-ANE algorithm for this particular $C$. Choose a running-time bound $S(N)=DN^d$ with $D\ge1$ and $d\ge1$. Since $b_0\ge1$,
\[
 S(M^{b_0})=DM^{b_0d}\ge M.
\]
Theorem~\ref{thm:hardness-amplification} therefore gives a $\delta$-ANE algorithm with running time
\[
 O\!\left(S(M^{b_0})^{a_0}\right)
 =O\!\left(D^{a_0}M^{a_0b_0d}\right),
\]
which is polynomial in $M$. This proves condition~\textnormal{(ii)}.

Conversely, assume condition~\textnormal{(ii)}. Fix $\eps\in(0,1)$, $\eta\in[0,1)$ and $C\ge2$. Use~\textnormal{(ii)} with $k=C$ and $\delta=\eps$ on symmetric binary-payoff games. The resulting $\eps$-ANE satisfies the approximation requirement for every player, so it is also an $(\eps,\eta)$-weak ANE. This proves~\textnormal{(i)} for every fixed pair $\eps,\eta$.
\end{proof}

\subsection{The PPAD-ETH consequence}\label{subsec:hardness}

We use deterministic PPAD-ETH in the convention of Rubinstein~\cite[Hypothesis~1]{DBLP:conf/focs/Rubinstein16}: solving \textsc{End-of-Line} requires deterministic time $2^{\widetilde\Omega(n)}$, where $n$ is the combined description length of the predecessor and successor circuits. We use the following consequence for bimatrix games.

\begin{theorem}[Rubinstein {\cite[Theorem~I.2]{DBLP:conf/focs/Rubinstein16}}]\label{thm:bimatrix}
Under deterministic PPAD-ETH, there is $\gamma_0\in(0,1)$ such that finding a $\gamma_0$-ANE of an $m\times m$ bimatrix game with payoffs in $[0,1]$ has no deterministic algorithm running in time $m^{O((\log m)^{1-\xi})}$ for any fixed $\xi\in(0,1)$. The hard instances have encoding length polynomial in $m$.
\end{theorem}

\begin{proof}[Proof of Corollary~\ref{cor:ppad-eth}]
Fix $\eps\in(0,1)$ and $\eta\in[0,1)$. Apply
Theorem~\ref{thm:hardness-amplification} with $k=2$ and
$\delta=\gamma_0$, obtaining constants $C,a_0,b_0$.
For the hard $m\times m$ input games from
Theorem~\ref{thm:bimatrix}, there is a constant $d_0\ge1$
such that their encoding length satisfies $M\le m^{d_0}$
for all sufficiently large $m$.

Suppose, for a fixed $\xi\in(0,1)$, that weak ANE in the constructed games could be found in time $N^{O((\log N)^{1-\xi})}$. We could then choose a nondecreasing running-time bound
\[
 S(N)=\exp\!\left(c(\log N)^{2-\xi}\right)
\]
for some constant $c>0$ and all sufficiently large $N$. Since $2-\xi>1$, this bound satisfies $S(N)\ge N$ for sufficiently large $N$, and hence $S(M^{b_0})\ge M$. By Theorem~\ref{thm:hardness-amplification}, there would be an algorithm for the input problem whose running time $R(M)$ satisfies, for some constant $D>0$,
\[
 \begin{aligned}
 R(M)&\le D S(M^{b_0})^{a_0}\\
 &=D\exp\!\left(a_0c(b_0\log M)^{2-\xi}\right)\\
 &\le D\exp\!\left(a_0c(b_0d_0\log m)^{2-\xi}\right)\\
 &=m^{O((\log m)^{1-\xi})}.
 \end{aligned}
\]
This contradicts Theorem~\ref{thm:bimatrix}. Since $\xi\in(0,1)$ was arbitrary, every running-time class excluded in Corollary~\ref{cor:ppad-eth} is excluded.
\end{proof}
\subsection{Average unilateral gain and bounds on the number of players}\label{subsec:upper}\label{subsec:small}
Given the parameters $\eps\in(0,1),\eta\in[0,1)$, 
let $\lambda=(1-\eta)(1-\eps)$. 
This quantity has a useful interpretation in terms of average unilateral gains. As shown in equation~\eqref{eq:weakANEaveragegain}, every $(\eps,\eta)$-weak ANE $\sigma$ of the constructed game satisfies
\begin{equation}\label{eq:sec4-averagegainlambda}
    \frac1C\sum_{i=1}^C
 \left[
   \max_{a_i\in A^t}U_i(a_i,\sigma_{-i})-U_i(\sigma)
 \right]
 \le(1-\eta)\eps+\eta
 =1-\lambda.
\end{equation}
That is, the average maximum unilateral gain is upper-bounded by $1-\lambda$. 
Note that this is exactly the condition required by our proof of Theorem~\ref{thm:amplification} to guarantee the success of the recovery calculation. The parameter choices~\eqref{eq:step-count} and~\eqref{eq:C-choice} ensure that
$q^t+\frac K{\sqrt C}<\lambda$.
By~\eqref{eq:failed-recovery-average-gain}, recovery therefore succeeds from every profile satisfying~\eqref{eq:sec4-averagegainlambda}. The proof uses the weak-ANE property only to establish this bound on the average maximum gain.

Next we study how the number of players needed for hardness amplification depends on the target parameter $\lambda\in(0,1)$. We first bound the number of players sufficient for our reduction. We then derive a necessary lower bound from existing polynomial-time approximation algorithms, which rule out the hardness conclusion for games with too few players.

\begin{proposition}[Number of players sufficient for the reduction]\label{prop:player-upper}
Fix $k\ge2$ and $\delta\in(0,1)$. For every $\eps\in(0,1)$ and $\eta\in[0,1)$, the reduction of Theorem~\ref{thm:amplification} can use a number of players satisfying
\begin{equation}\label{eq:player-upper}
 C\le\exp\left(O_{k,\delta}\bigl([1+\log(\tfrac1\lambda)]^2\bigr)\right),
 \qquad \lambda=(1-\eta)(1-\eps).
\end{equation}
The constant implicit in $O_{k,\delta}$ depends only on $k,\delta$.
\end{proposition}

\begin{proof}
Fix $k,\delta$, and choose $\beta,\Delta$, and hence $q,h$, as functions only of $k,\delta$. The choice of $t$ in~\eqref{eq:step-count} gives
\[
 \frac{\log(2/\lambda)}{\log(1/q)}
 \le t<1+\frac{\log(2/\lambda)}{\log(1/q)},
 \qquad
 t=O_{k,\delta}(1+\log(1/\lambda)).
\]
Appendix~\ref{app:constants} proves that $\log(K+1)\le Bt^2$ for a constant $B$ depending only on $k,\delta$. Choose an integer $C$ with
\begin{equation}\label{eq:C-integer}
 \left(\frac{2K}{\lambda}\right)^2+2
 \le C\le
 \left(\frac{2K}{\lambda}\right)^2+3.
\end{equation}
Such an integer exists because the interval has length one. It satisfies $C\ge2$ and $\sqrt C>2K/\lambda$, as required by~\eqref{eq:C-choice}. Moreover,
\[
 C\le3+\frac{4K^2}{\lambda^2}
 \le\frac{7(K+1)^2}{\lambda^2}.
\]
Taking logarithms gives
\[
 \begin{aligned}
 \log C
 &\le\log7+2\log(K+1)+2\log\frac1\lambda\\
 &\le\log7+2Bt^2+2\log\frac1\lambda\\
 &=O_{k,\delta}\bigl([1+\log(1/\lambda)]^2\bigr).
 \end{aligned}
\]
Exponentiating proves~\eqref{eq:player-upper}.
\end{proof}

We can also derive a necessary lower bound on the number of players for the hardness of approximation. Briest, Goldberg, and R\"oglin~\cite[Theorem~5]{DBLP:journals/corr/abs-0804-4524} give a deterministic polynomial-time algorithm for computing a $(1-1/p)$-ANE for every fixed number $p\ge2$ of players. Consider applying this algorithm to a game formed by selecting a subset of players and fixing the other players' actions. The following proposition identifies when the resulting profile is an $(\eps,\eta)$-weak ANE of the original game.

\begin{proposition}[Polynomial-time weak ANE with few players]\label{prop:small}
Fix an integer $p\ge1$, $\eps\in(0,1)$ and $\eta\in[0,1)$. If
\begin{equation}\label{eq:selected-count}
 \left\lceil(1-\eta)p\right\rceil\le\frac1{1-\eps},
\end{equation}
then an $(\eps,\eta)$-weak ANE can be found deterministically in polynomial time for every $p$-player game with payoffs in $[0,1]$.
\end{proposition}

\begin{proof}
Select a set $I$ of $\lceil(1-\eta)p\rceil$ players, and fix an arbitrary pure action for each player outside $I$. Form a game on $I$ by evaluating the selected players' payoff functions with those actions fixed. If $|I|=1$, choose a best response to the fixed actions. Otherwise, compute a $(1-1/|I|)$-ANE of the game on $I$ using the algorithm in \cite[Theorem~5]{DBLP:journals/corr/abs-0804-4524}. Return the computed strategies together with the fixed actions of the players outside $I$.

Let $\widetilde\sigma$ denote the computed profile in the game on $I$, let $\widetilde U_i$ denote its expected-payoff functions, and let $\sigma$ be the returned profile in the original game. Since the players outside $I$ use the actions fixed when the smaller game was formed, every selected player $i$ and every action $a_i$ satisfy
\[
 U_i(a_i,\sigma_{-i})
 =\widetilde U_i(a_i,\widetilde\sigma_{-i}),
 \qquad
 U_i(\sigma)=\widetilde U_i(\widetilde\sigma).
\]
Thus, for every $i\in I$,
\[
 \begin{aligned}
 &\max_{a_i\in A_i}U_i(a_i,\sigma_{-i})-U_i(\sigma)\\
 &\quad=\max_{a_i\in A_i}\widetilde U_i(a_i,\widetilde\sigma_{-i})-\widetilde U_i(\widetilde\sigma)\\
 &\quad\le1-\frac1{|I|}
 =1-\frac1{\lceil(1-\eta)p\rceil}
 \le\eps,
 \end{aligned}
\]
where the last inequality follows from~\eqref{eq:selected-count}. The number of players violating the approximation requirement is therefore at most
\[
 p-|I|=p-\lceil(1-\eta)p\rceil=\lfloor\eta p\rfloor.
\]
The game on $I$ is obtained by taking entries from the original payoff table. Since $|I|\le p$ and $p$ is fixed, constructing this game, solving it, and returning the full profile take polynomial time in the original input length.
\end{proof}

Proposition~\ref{prop:small} implies that the conditional hardness conclusion of Corollary~\ref{cor:ppad-eth} can hold for a number $p$ of players only if
\begin{equation}\label{eq:player-necessary-uniform}
 p>\frac1{2(1-\eta)(1-\eps)}=\frac1{2\lambda}.
\end{equation}
Indeed, suppose $p\le1/(2\lambda)$. If $(1-\eta)p\le1$, then
$
 \lceil(1-\eta)p\rceil=1\le\frac1{1-\eps}$.
If $(1-\eta)p>1$, then
$
 \lceil(1-\eta)p\rceil
 <(1-\eta)p+1
 \le2(1-\eta)p
 \le\frac1{1-\eps}$.
In either case, condition~\eqref{eq:selected-count} holds. Proposition~\ref{prop:small} then gives a polynomial-time algorithm, contradicting the hardness conclusion.

For the hardness result of Corollary~\ref{cor:ppad-eth}, applying Proposition~\ref{prop:player-upper} with fixed $k=2$ and $\delta=\gamma_0$ to the hard bimatrix games supplied by Theorem~\ref{thm:bimatrix} gives a sufficient number of players quasipolynomial in $1/\lambda$. On the other hand, \eqref{eq:player-necessary-uniform} gives a necessary lower bound linear in $1/\lambda$. 
These leave open whether the sufficient number of players for the reduction can be reduced to $(1/\lambda)^{O_{k,\delta}(1)}$, or even to $O_{k,\delta}(1/\lambda)$.

\newpage
\appendix
\section*{Appendix}
\section{Proof of the sampling bound for functions of several sequences}\label{app:product-sampling}

Recall the setting of Lemma~\ref{lem:product-sampling}: $P=\frac1C\sum_{j=1}^C\sigma_j$ is a mixture of distributions on $A^t$, $I\subseteq[C]$ satisfies $|I|\le\ell$, and $\sigma_{-I}=\bigotimes_{j\notin I}\sigma_j$ is the law of the independent draws outside $I$.

\begin{proof}[Proof of Lemma~\ref{lem:product-sampling}]
Keep $f$ and the sequences indexed by $I$ fixed. All expectations,
variances, and covariances involving $Y_{-I}$ below are taken under
$Y_{-I}\sim\sigma_{-I}$.

For each ordered index tuple
$\mathbf j=(j_1,\ldots,j_{k-1})\in[C]^{k-1}$, define
\[
 X_{\mathbf j}=f(Y_{j_1},\ldots,Y_{j_{k-1}}),
 \qquad
 T=\frac1{C^{k-1}}\sum_{\mathbf j\in[C]^{k-1}}X_{\mathbf j}.
\]
Write $\mathbf z=(z^{(1)},\ldots,z^{(k-1)})$ and define
\[
 \mu_{\mathbf j}
 =\sum_{\mathbf z\in(A^t)^{k-1}}
      f(\mathbf z)\prod_{r=1}^{k-1}\sigma_{j_r}(z^{(r)}).
\]
Thus $\mu_{\mathbf j}$ uses independent draws in all $k-1$
positions, including when some indices in $\mathbf j$ coincide.
Expanding $P=\frac1C\sum_{j=1}^C\sigma_j$ gives
\[
\begin{aligned}
 \mu
 &:=\E_{\mathbf Z\sim P^{\otimes(k-1)}}[f(\mathbf Z)]\\
 &=\sum_{\mathbf z\in(A^t)^{k-1}}
       f(\mathbf z)\prod_{r=1}^{k-1}P(z^{(r)})\\
 &=\sum_{\mathbf z\in(A^t)^{k-1}}
       f(\mathbf z)
       \prod_{r=1}^{k-1}
       \left(\frac1C\sum_{j_r=1}^C
                         \sigma_{j_r}(z^{(r)})\right)\\
 &=\frac1{C^{k-1}}\sum_{\mathbf j\in[C]^{k-1}}
       \sum_{\mathbf z\in(A^t)^{k-1}}
          f(\mathbf z)\prod_{r=1}^{k-1}
                              \sigma_{j_r}(z^{(r)})\\
 &=\frac1{C^{k-1}}\sum_{\mathbf j\in[C]^{k-1}}
       \mu_{\mathbf j}.
\end{aligned}
\]

We first bound $|\E_{Y_{-I}}[T]-\mu|$.
Let $\mathcal B$ contain the index tuples that use an index in $I$
or repeat an index:
\[
 \mathcal B
 =
 \bigcup_{r=1}^{k-1}
   \{\mathbf j:j_r\in I\}
 \;\cup\!
 \bigcup_{1\le r<s\le k-1}
   \{\mathbf j:j_r=j_s\}.
\]
If $\mathbf j\notin\mathcal B$, then
$Y_{j_1},\ldots,Y_{j_{k-1}}$ are independent with respective laws
$\sigma_{j_1},\ldots,\sigma_{j_{k-1}}$. Hence
\[
 \E_{Y_{-I}}[X_{\mathbf j}]
 =
 \sum_{\mathbf z\in(A^t)^{k-1}}
      f(\mathbf z)\prod_{r=1}^{k-1}\sigma_{j_r}(z^{(r)})
 =\mu_{\mathbf j}.
\]
For every $\mathbf j$, the bound $0\le f\le1$ gives
\[
 0\le\E_{Y_{-I}}[X_{\mathbf j}],\mu_{\mathbf j}\le1,
 \qquad
 \bigl|\E_{Y_{-I}}[X_{\mathbf j}]-\mu_{\mathbf j}\bigr|\le1.
\]

For each position $r$, there are $|I|C^{k-2}$ index tuples
with $j_r\in I$. For each pair of positions $r<s$, there are
$C^{k-2}$ index tuples with $j_r=j_s$. Therefore
\[
\begin{aligned}
 |\mathcal B|
 &\le
 \sum_{r=1}^{k-1}|\{\mathbf j:j_r\in I\}|
 +\sum_{1\le r<s\le k-1}|\{\mathbf j:j_r=j_s\}|\\
 &=\left((k-1)|I|+\binom{k-1}{2}\right)C^{k-2}\\
 &\le\left(k|I|+\binom{k}{2}\right)C^{k-2}.
\end{aligned}
\]
Consequently,
\begin{equation}\label{eq:product-bias}
\begin{aligned}
 \bigl|\E_{Y_{-I}}[T]-\mu\bigr|
 &=\left|
   \frac1{C^{k-1}}\sum_{\mathbf j\in\mathcal B}
      \bigl(\E_{Y_{-I}}[X_{\mathbf j}]-\mu_{\mathbf j}\bigr)
   \right|\\
 &\le\frac1{C^{k-1}}\sum_{\mathbf j\in\mathcal B}
      \bigl|\E_{Y_{-I}}[X_{\mathbf j}]-\mu_{\mathbf j}\bigr|\\
 &\le\frac{|\mathcal B|}{C^{k-1}}
 \le\frac{k|I|+\binom{k}{2}}{C}.
\end{aligned}
\end{equation}

We next bound $\operatorname{Var}_{Y_{-I}}(T)$.
Let $\mathcal D$ contain the ordered pairs of index tuples
that share at least one index outside $I$:
\[
 \mathcal D
 =
 \bigcup_{r,s=1}^{k-1}\;
 \bigcup_{a\in[C]\setminus I}
 \{(\mathbf j,\mathbf j'):j_r=j'_s=a\}.
\]
If $(\mathbf j,\mathbf j')\notin\mathcal D$, then
$X_{\mathbf j}$ and $X_{\mathbf j'}$ depend on disjoint
collections of the independent sequences indexed outside $I$.
The sequences indexed by $I$ are fixed. Thus
\[
 \E_{Y_{-I}}[X_{\mathbf j}X_{\mathbf j'}]
 =
 \E_{Y_{-I}}[X_{\mathbf j}]
 \E_{Y_{-I}}[X_{\mathbf j'}],
 \qquad
 \operatorname{Cov}_{Y_{-I}}
       (X_{\mathbf j},X_{\mathbf j'})=0.
\]

For every $\mathbf j$, the bound $0\le X_{\mathbf j}\le1$ gives
\[
\begin{aligned}
 \operatorname{Var}_{Y_{-I}}(X_{\mathbf j})
 &=
 \E_{Y_{-I}}\!\left[
       \left(X_{\mathbf j}-\frac12\right)^2\right]
 -\left(\E_{Y_{-I}}[X_{\mathbf j}]-\frac12\right)^2\\
 &\le
 \E_{Y_{-I}}\!\left[
       \left(X_{\mathbf j}-\frac12\right)^2\right]
 \le\frac14.
\end{aligned}
\]
Cauchy--Schwarz therefore yields
\[
 \left|\operatorname{Cov}_{Y_{-I}}
       (X_{\mathbf j},X_{\mathbf j'})\right|
 \le
 \sqrt{\operatorname{Var}_{Y_{-I}}(X_{\mathbf j})
       \operatorname{Var}_{Y_{-I}}(X_{\mathbf j'})}
 \le\frac14.
\]

For fixed positions $r,s$ and a fixed shared index $a\notin I$,
the conditions $j_r=j'_s=a$ leave $2(k-1)-2$ coordinates free.
Hence
\[
\begin{aligned}
 |\mathcal D|
 &\le
 \sum_{r,s=1}^{k-1}
 \sum_{a\in[C]\setminus I}
   |\{(\mathbf j,\mathbf j'):j_r=j'_s=a\}|\\
 &=(k-1)^2(C-|I|)C^{2(k-1)-2}.
\end{aligned}
\]
Expanding the variance of $T$ now gives
\begin{equation}\label{eq:product-variance}
\begin{aligned}
 \operatorname{Var}_{Y_{-I}}(T)
 &=
 \frac1{C^{2(k-1)}}
 \sum_{\mathbf j,\mathbf j'\in[C]^{k-1}}
   \operatorname{Cov}_{Y_{-I}}
       (X_{\mathbf j},X_{\mathbf j'})\\
 &=
 \frac1{C^{2(k-1)}}
 \sum_{(\mathbf j,\mathbf j')\in\mathcal D}
   \operatorname{Cov}_{Y_{-I}}
       (X_{\mathbf j},X_{\mathbf j'})\\
 &\le
 \frac1{C^{2(k-1)}}
 \sum_{(\mathbf j,\mathbf j')\in\mathcal D}
   \left|\operatorname{Cov}_{Y_{-I}}
       (X_{\mathbf j},X_{\mathbf j'})\right|\\
 &\le\frac{|\mathcal D|}{4C^{2(k-1)}}\\
 &\le\frac{(k-1)^2(C-|I|)}{4C^2}
 \le\frac{k^2}{4C}.
\end{aligned}
\end{equation}

Finally, the triangle inequality, Cauchy--Schwarz,
\eqref{eq:product-variance}, and \eqref{eq:product-bias} give
\[
\begin{aligned}
 \E_{Y_{-I}}[|T-\mu|]
 &\le
 \E_{Y_{-I}}\!\left[
       |T-\E_{Y_{-I}}[T]|\right]
 +|\E_{Y_{-I}}[T]-\mu|\\
 &\leq\sqrt{\operatorname{Var}_{Y_{-I}}(T)}
   +|\E_{Y_{-I}}[T]-\mu|\\
 &\le\frac{k}{2\sqrt C}
       +\frac{k|I|+\binom{k}{2}}{C}\\
 &\le\frac{k}{2\sqrt C}
       +\frac{k\ell+\binom{k}{2}}{C}.
\end{aligned}
\]
This is~\eqref{eq:product-sampling}.
\end{proof}

\section{Exact threshold computation and bit complexity}\label{app:arithmetic}

We first compute one filtering round exactly, then bound the arithmetic needed for construction and recovery. Throughout this appendix, $\rho,w,v$ are functions on $A^t$, as in the definition of $\Filter$. Write $n_A=|A|\ge1$, so each function has $n_A^t$ entries. The parameters $k,C,t,\beta,\Delta,q,h$ are fixed, with $0<q<1$ and $h>0$. The payoff table of $G$ and supplied strategy probabilities are input data. All encoding lengths below are measured in bits.

\begin{lemma}[Exact computation of one filtering round]\label{lem:exact-filter}
Suppose rational values $\rho(z)\ge0$ and $w(z),v(z)\in[0,1]$ are given for every $z\in A^t$, with
$
m=\sum_{z\in A^t}\rho(z)w(z)>0$.
The smallest threshold and all updated weights defined by $\Filter(\rho,w,v)$ can be computed exactly in polynomial time in the input encoding length. More precisely, if $\ell_{\mathrm{in}}\ge1$ bounds the length of every input rational, including $q$ and $h$, then each output rational has $O(n_A^t(\ell_{\mathrm{in}}+1))$ bits, and the total output length is $O(n_A^{2t}(\ell_{\mathrm{in}}+1))$.
\end{lemma}

\begin{proof}
For a trial threshold $\theta$, the remaining mass is
\begin{equation}\label{eq:generic-threshold}
 F(\theta)=\sum_{z\in A^t}\rho(z)w(z)\phi(v(z),\theta).
\end{equation}
We must find the smallest $\theta\in[-h,1]$ satisfying $F(\theta)=qm$. The function $F$ is continuous and nonincreasing, and
\[
 F(-h)=m>qm>0=F(1).
\]
The endpoint values use $v(z)\in[0,1]$, and the strict inequalities use $m>0$ and $0<q<1$. Thus a threshold satisfying the equation exists.

For each sequence $z$, the formula for $\phi(v(z),\theta)$ changes only at $v(z)-h$ and $v(z)$. Sort these breakpoints together with $-h$ and $1$, removing duplicates. On an open interval $I$ between consecutive breakpoints, let $S_1(I)$ consist of the sequences whose multipliers equal one, and let $S_{\mathrm{lin}}(I)$ consist of those whose multipliers equal $(v(z)-\theta)/h$. These sets stay fixed as $\theta$ varies within $I$. Define
\begin{equation}\label{eq:affine-threshold}
 \begin{aligned}
 a_I&=\sum_{z\in S_1(I)}\rho(z)w(z)
       +\frac1h\sum_{z\in S_{\mathrm{lin}}(I)}\rho(z)w(z)v(z),\\
 b_I&=\frac1h\sum_{z\in S_{\mathrm{lin}}(I)}\rho(z)w(z),\\
 F(\theta)&=a_I-b_I\theta.
 \end{aligned}
\end{equation}
The last equality substitutes the two nonzero multiplier formulas into \eqref{eq:generic-threshold}. Nonnegativity of the masses and weights gives $b_I\ge0$, and continuity extends the affine formula to both endpoints of $I$.

If $b_I>0$, the only possible solution on this interval is
\begin{equation}\label{eq:threshold-solution}
 \theta_I=\frac{a_I-qm}{b_I}.
\end{equation}
Compute this rational number and test whether it lies between the two endpoints, inclusive. If $b_I=0$, compare $a_I$ with $qm$. If they are equal, the entire closed interval consists of solutions and its smallest point is the left endpoint. Otherwise there is no solution on the interval. Examine the intervals in increasing order and return the first solution found. This procedure also handles solutions at breakpoints and flat intervals of solutions, so it returns the smallest threshold required by $\Filter$.

With this threshold, compute
\[
 w^+(z)=w(z)\phi(v(z),\theta)\qquad(\forall\,z\in A^t).
\]
This produces a weight at every sequence, including sequences with zero initial mass.

To bound the encoding lengths, represent each rational by an integer numerator and a positive integer denominator. Products and nonzero quotients have lengths bounded by the sum of the operand lengths, up to an absolute constant. A sum of at most $n_A^t$ rationals, each of length $O(\ell_{\mathrm{in}}+1)$, can use the product of their denominators as a common denominator. The resulting numerator and denominator have $O(n_A^t(\ell_{\mathrm{in}}+1))$ bits. Reduction to lowest terms can only shorten them.

The products $\rho(z)w(z)$ and $\rho(z)w(z)v(z)$ have $O(\ell_{\mathrm{in}}+1)$ bits. The sums defining $m,a_I,b_I$, the quotient in \eqref{eq:threshold-solution}, and each updated weight therefore have $O(n_A^t(\ell_{\mathrm{in}}+1))$ bits. There are $n_A^t+1$ output rationals, giving the stated total output length.

There are $O(n_A^t)$ breakpoints and intervals. Sorting requires $O(n_A^t\log(n_A^t+1))$ rational comparisons. On each interval, recompute $a_I,b_I$ from the original input values in $O(n_A^t)$ rational operations. The total is $O(n_A^{2t})$ operations, with intermediate lengths bounded as above. Exact integer arithmetic therefore implements the entire procedure in polynomial bit time.
\end{proof}

\paragraph{Arithmetic across rounds.}
For both the payoff and recovery calculations, we enumerate all sequences in $A^t$ and compute their weights in every round. The initial masses remain fixed, the weights start at one, and the total mass entering round $s$ is $m_s=q^{s-1}$. In this arithmetic estimate, write $\rho,w$ for the current inputs: they are $(\widehat P,\widehat W_{s-1})$ in the payoff calculation and $(P,W_{s-1})$ in the recovery calculation.

Compute the distribution $Q$ of current tuples and its marginals $x_r$ by grouping the masses according to the current coordinate:
\begin{equation}\label{eq:arithmetic-marginals}
 \begin{aligned}
 Q(a)&=\frac1{m_s}\sum_{\substack{z\in A^t\\z_s=a}}
                   \rho(z)w(z),\\
 x_r(a_r)&=\frac1{m_s}\sum_{\substack{z\in A^t\\(z_s)_r=a_r}}
                   \rho(z)w(z).
 \end{aligned}
\end{equation}
Here $a\in A$, $a_r\in A_r$, and $(z_s)_r$ is the action of player $r$ in the tuple $z_s$. One enumeration of $A^t$ computes these sums for every action tuple and every action in each marginal distribution. For each player $r$ of $G$, compute the payoff of each action against the product of the opponent marginals, then average the payoffs to obtain the scores:
\begin{equation}\label{eq:arithmetic-scores}
 \begin{aligned}
 u_r(a_r,x_{-r})
 &=\sum_{a_{-r}\in\prod_{\ell\ne r}A_\ell}
       u_r(a_r,a_{-r})\prod_{\ell\ne r}x_\ell(a_\ell),\\
 V_Q(a)&=\frac1k\sum_{r=1}^k u_r(a_r,x_{-r}).
 \end{aligned}
\end{equation}
For each $r$, one enumeration of the payoff table of $G$ computes $u_r(a_r,x_{-r})$ for every $a_r\in A_r$. Each term uses $k-1$ marginal probabilities. Then compute $V_Q(a)$ for every $a\in A$ and assign $v(z)=V_Q(z_s)$ for every $z\in A^t$. The algorithm in Lemma~\ref{lem:exact-filter} then computes
the threshold and updated weights. For fixed $k$, these steps use polynomially many rational operations in $n_A^t$ and $n_A$.

We now track the lengths of these rationals. Let $\ell_{\mathrm{base}}\ge1$ bound the lengths of the initial masses, payoffs of $G$, $h$, and powers $q^j$ for $0\le j\le t$. Let $\ell_s$ be the maximum encoding length of a weight after round $s$, and put $b_s=\ell_{s-1}+\ell_{\mathrm{base}}+1$. A current mass $\rho(z)w(z)$ has $O(b_s)$ bits. Each probability in \eqref{eq:arithmetic-marginals} sums at most $n_A^t$ such masses and divides by $m_s$, so its length is $O(n_A^t b_s)$.

Each product in \eqref{eq:arithmetic-scores} contains one payoff of $G$ and $k-1$ marginal probabilities, giving length $O_k(n_A^t b_s)$. Each payoff sum has at most $n_A$ terms, and the score averages $k$ such sums. Thus the scores have $O_k(n_A^{t+1}b_s)$ bits. Applying the algorithm in the proof of
Lemma~\ref{lem:exact-filter} with these assigned values
gives updated weights and intermediate quantities of
$O_k(n_A^{2t+1}b_s)$ bits.

Choose a sufficiently large constant $c_k\ge2$ and set $\Pi=c_k(n_A+1)^{2t+1}$. The preceding bounds give
\begin{equation}\label{eq:round-bit-length}
 \begin{aligned}
 \ell_s&\le\Pi(\ell_{s-1}+\ell_{\mathrm{base}}+1),\\
 \ell_s+\ell_{\mathrm{base}}+1
 &\le(\Pi+1)(\ell_{s-1}+\ell_{\mathrm{base}}+1)\\
 &\le(\Pi+1)^s(\ell_0+\ell_{\mathrm{base}}+1).
 \end{aligned}
\end{equation}
Iterating the inequality for $\ell_s+\ell_{\mathrm{base}}+1$ gives the last line. The initial length $\ell_0$ is constant because every initial weight is one. Since $s\le t$ and $k,t$ are fixed, the resulting lengths are polynomial in $n_A$ and $\ell_{\mathrm{base}}$. Each round also uses polynomially many rational operations, so all $t$ rounds take polynomial bit time.

\paragraph{Construction of the complete table.}
Let $M$ be the encoding length of the input game. Its complete payoff table gives $n_A=|A|=O(M)$. For a pure profile of $H$, form $\widehat P(z)=n_z/C$ for every $z\in A^t$, where $n_z$ is the number of players choosing $z$. The numerators and denominators have lengths bounded in terms of the fixed number of players, and the initial weights are one. Thus $\ell_{\mathrm{base}}$ is polynomially bounded in $M$.

Compute all $t$ rounds exactly on $A^t$, and test whether each player's final weight equals one to determine its binary payoff. The preceding bounds give polynomial time in $M$ per pure profile. There are $n_A^{tC}$ pure profiles, so enumerating them and writing the $C$ payoffs per profile also takes polynomial time, because $C,t$ are fixed. In particular, the complete payoff-table encoding length of $H$ satisfies $N\le M^{b_0}$ for a sufficiently large fixed integer $b_0$, after absorbing constants for $M\ge2$.

\paragraph{Recovery from a supplied profile.}
Let $L_\sigma$ be the actual encoding length of the supplied mixed profile. For every $z\in A^t$, compute
\[
 P(z)=\frac1C\sum_{j=1}^C\sigma_j(z).
\]
Each probability is a finite rational sum with length polynomial in $L_\sigma$, so $\ell_{\mathrm{base}}$ is polynomially bounded in $M+L_\sigma$. The recovery calculation enumerates all $n_A^t$ sequences and $n_A$ action tuples in each round. Since $t$ is fixed and $n_A=O(M)$, the preceding bounds show that all distributions, scores, thresholds and weights are computed in polynomial time and have polynomial encoding length.

After computing all $t$ rounds, evaluate the score test for every $Q_s$. Compute $V_s^*=\max_{a\in A}V_{Q_s}(a)$ by exact comparisons, sum the $Q_s$-mass of the tuples with $V_s^*-V_{Q_s}(a)\le\Delta$, and compare it with $1-\beta$. If a score test passes, select the smallest such round index and compute the $k$ marginals of its distribution by finite sums. By Lemma~\ref{lem:criterion} and the parameter choice \eqref{eq:source-accuracy}, these marginals form a
$\delta$-ANE of $G$. The entire computation and its output length are polynomial in $M+L_\sigma$, including when every score test fails. This proves the computational claims of Theorem~\ref{thm:amplification} and supplies the exponents used in Theorem~\ref{thm:hardness-amplification}.

\section{Explicit error constants and their growth}\label{app:constants}

We will prove that the deviation constant $K$ from~\eqref{eq:K} satisfies $\log(K+1)=O_{k,\delta}(t^2)$ when $\beta,\Delta$ are chosen as functions only of $k,\delta$. We first verify the coefficient comparison used in the proof of Lemma~\ref{lem:deviation}, then bound the growth of $K$ using \eqref{eq:main-error-coefficients}. Throughout this appendix, $k,\beta,\Delta$, and hence $\Gamma_k=3k^2$, $q,h$, remain fixed while $t$ varies.

\paragraph{The coefficient in the payoff bound.}
For $1\le s\le t$, we have $m_s=q^{s-1}\le1$ and $h=\Delta/4$. The coefficient in~\eqref{eq:round-failure} therefore satisfies
\begin{equation}\label{eq:K-round-bound}
 \begin{aligned}
 &\frac2\beta\left[
       \frac{E_{s-1}^{W}+2}{m_s}
       +\frac{E_s^{V}}{3h}\left(\frac1{m_s}+1\right)\right]\\
 &\quad\le\frac{2(E_{s-1}^{W}+2)}{\beta m_s}
       +\frac{16E_s^{V}}{3\beta\Delta m_s}\\
 &\quad\le\frac{4(E_{s-1}^{W}+2+4E_s^{V}/\Delta)}{\beta m_s}
       +\frac{4E_s^{V}}{\Delta}.
 \end{aligned}
\end{equation}
The first inequality uses $m_s\le1$ and $h=\Delta/4$. The second enlarges nonnegative terms. Its right side is the $s$-th summand of $K$ after substituting $m_s=q^{s-1}$.

For the growth estimate below, we substitute $m_s=q^{s-1}$ in the error recurrence to make the dependence on the round index explicit.

\paragraph{Growth with the number of rounds.}
Set $F_s=E_s^{W}+1$, so $F_0=1$ and $F_s\ge1$, and put $A_0=3+2\Gamma_k/h>1$. Since $\Gamma_k\ge k-1$, the recurrence for $E_s^{V}$ gives
\[
 E_s^{V}=\frac{(k-1)E_{s-1}^{W}+\Gamma_k}{q^{(k-1)(s-1)}}
 \le\frac{\Gamma_kF_{s-1}}{q^{(k-1)(s-1)}}.
\]
Substituting this bound into the recurrence for $E_s^{W}$ yields
\[
 \begin{aligned}
 F_s
 &=2F_{s-1}+1+\frac{2E_s^{V}}{h}\\
 &\le\left(3+\frac{2\Gamma_k}{hq^{(k-1)(s-1)}}\right)F_{s-1}\\
 &\le\frac{A_0}{q^{(k-1)(s-1)}}F_{s-1}.
 \end{aligned}
\]
Here $F_{s-1}\ge1$ and $q^{(k-1)(s-1)}\le1$. Multiplying the bounds through round $s$, then substituting into the estimate for $E_s^{V}$, gives
\begin{equation}\label{eq:R-growth}
 \begin{aligned}
 F_s
 &\le\prod_{v=1}^s\frac{A_0}{q^{(k-1)(v-1)}}
 =\frac{A_0^s}{q^{(k-1) s(s-1)/2}},\\
 E_s^{V}
 &\le\frac{\Gamma_kF_{s-1}}{q^{(k-1)(s-1)}}
 \le\frac{\Gamma_kA_0^{s-1}}{q^{(k-1)(s-1)(s-2)/2+(k-1)(s-1)}}
 =\frac{\Gamma_kA_0^{s-1}}{q^{(k-1) s(s-1)/2}}.
 \end{aligned}
\end{equation}
The exponents are the sums of the normalization exponents from the preceding rounds.

Separate the three contributions to the round-$s$ summand in $K$. Since $E_{s-1}^{W}+2=F_{s-1}+1\le2F_{s-1}$, they satisfy
\[
 \begin{aligned}
 \frac{4(E_{s-1}^{W}+2)}{\beta q^{s-1}}
 &\le\frac{8F_{s-1}}{\beta q^{s-1}}
 \le\frac{8A_0^{s-1}}{\beta q^{(k-1) s(s-1)/2}},\\
 \frac{16E_s^{V}}{\beta\Delta q^{s-1}}
 &\le\frac{16\Gamma_kA_0^{s-1}}
           {\beta\Delta q^{(k-1) s(s-1)/2+s-1}},\\
 \frac{4E_s^{V}}{\Delta}
 &\le\frac{4\Gamma_kA_0^{s-1}}{\Delta q^{(k-1) s(s-1)/2}}.
 \end{aligned}
\]
The first line uses~\eqref{eq:R-growth} and
$(k-1)(s-1)(s-2)/2+s-1\le (k-1) s(s-1)/2$, which follows from $k\ge2$. The other lines substitute the bound on $E_s^{V}$. For $1\le s\le t$, all powers of $A_0$ are at most $A_0^t$ and all exponents of $1/q$ are at most $(k-1) t(t-1)/2+t$. Define
\[
 D_0=\frac8\beta+\frac{16\Gamma_k}{\beta\Delta}+\frac{4\Gamma_k}{\Delta}>1.
\]
Summing the resulting common bound over all rounds gives
\begin{equation}\label{eq:K-growth}
 K\le\frac{D_0tA_0^t}{q^{(k-1) t(t-1)/2+t}}.
\end{equation}
The right side is at least one, so $K+1$ is at most twice that quantity. Thus
\[
 \begin{aligned}
 \log(K+1)
 &\le\log(2D_0)+\log t+t\log A_0
      +\left(\frac{(k-1) t(t-1)}2+t\right)\log\frac1q\\
 &=O_{k,\delta}(t^2).
 \end{aligned}
\]
For the last equality, choose $\beta,\Delta$ as functions only of $k,\delta$. Then $D_0,A_0,q$ depend only on the fixed parameters $k,\delta$. This is the bound used in Proposition~\ref{prop:player-upper}.

\subsection*{AI methodology}
We used OpenAI's GPT 6 Astra and GPT 6.1 Sol to help find and simplify the proof of the main result. The tool materially affected all sections in this paper. The authors verified the correctness and originality of all content, including references.

\begingroup
\small
\bibliographystyle{alphaurl}
\bibliography{refs}
\endgroup

\end{document}